\documentclass[a4paper, 11pt, twoside]{article}

\usepackage[a4paper,top=3cm,bottom=2cm,left=3cm,right=3cm,marginparwidth=1.75cm]{geometry}

\usepackage{graphicx}
\usepackage{amsmath}
\usepackage{amsfonts}
\usepackage{amsthm}

\usepackage{nicematrix}

\usepackage{subcaption}

\usepackage{tikz} 
\usetikzlibrary{arrows}
\usetikzlibrary{decorations.pathmorphing}
\usetikzlibrary{decorations.markings}
\usetikzlibrary{patterns}
\usetikzlibrary{automata}
\usetikzlibrary{positioning}
\usepackage{tikz-cd}
\tikzset{->-/.style={decoration={
				markings,
				mark=at position #1 with {\arrow{latex}}},postaction={decorate}}}
	
	\tikzset{-<-/.style={decoration={
				markings,
				mark=at position #1 with {\arrowreversed{latex}}},postaction={decorate}}}

\usetikzlibrary{shapes.misc}\tikzset{cross/.style={cross out, draw, 
         minimum size=2*(#1-\pgflinewidth), 
         inner sep=0pt, outer sep=0pt}}

\usepackage{pgfplots}
\usepackage{appendix}

\newtheorem{theorem}{Theorem}[section]

\newtheorem{lemma}[theorem]{Lemma}

\newtheorem{remark}[theorem]{Remark}
\newtheorem{problem}[theorem]{RH problem}

\newcommand{\Z}{\mathbb{Z}}

\newcommand{\C}{\mathbb{C}}
\newcommand{\F}{\mathcal{F}}

\usepackage[numbers, comma, square, sort&compress]{natbib}
\numberwithin{equation}{section}

\title{Reducing the domain of discrete multi-time determinantal point processes}
\author{Tom Claeys and Felix Gideonse}

\begin{document}

\maketitle

%\tableofcontents
\begin{abstract}
We consider a class of discrete multi-time determinantal point processes whose correlation kernels admit a double-contour integral form, containing Schur processes as simple examples. We show that the form of the correlation kernel is preserved under conditioning of the point process to a restricted domain, however, with an integrand which becomes more complicated and which is characterized by a Riemann-Hilbert problem. We apply our general result to domino tilings of reduced Aztec diamonds and to lozenge tilings of hexagons with holes. Our results show a striking analogy with the Its-Izergin-Korepin-Slavnov method for integrable kernels.
\end{abstract}

\section{Introduction} \label{intro}

We consider a discrete lattice $\Lambda_N$ consisting of a finite number of copies of the integers,
\begin{equation}\label{def:lattice}
\Lambda_N:=\{1,2,\dots, N\}\times\mathbb Z=\{(t,h): t\in\{1,2,\dots, N\}, h\in\mathbb Z\}.
\end{equation}
For a point $(t,h)$ on the lattice, we will interpret $t$ as a (discrete) time variable and $h$ as a (discrete) space variable.
We consider determinantal point processes $\mathbb{P}$ on $\Lambda_N$, which are probability distributions on the space of point configurations $\mathcal{X}$ in $\Lambda_N$, whose correlation functions have a determinantal form in the following sense:
there exists a correlation kernel
\[K:\Lambda_N\times\Lambda_N\to \mathbb C:\left(t,h;t',h'\right)\mapsto K\left(t,h;t',h'\right),\]
which is such that, for every $m\in\mathbb N$, and for every collection of $m$ distinct points $(t_1,h_1),\ldots, (t_m,h_m) \in \Lambda_N$, we have that
\begin{equation}
    \mathbb P\left((t_j,h_j)\in \mathcal X\ \mbox{for every }j=1,\ldots, m\right)
=\det\left[K\left(t_j,h_j;t_k,h_k\right)\right]_{j,k=1}^m.
\end{equation}
We refer the reader to  \cite{Soshnikov_2000, borodin2015dpp, JOHANSSON2006dpp} for background and general theory of determinantal point processes.

\medskip

We assume that the correlation kernel belongs to a general class of kernels which we introduce now. Suppose that $\vec W = (W_1, \dots, W_N)$ is a sequence of complex-valued functions such that, for $t, t' = 1, \dots, N$, the function $1_{t>t'} W_{t'}/W_t$ is analytic in an annular neighbourhood of the unit circle in the complex plane, $\mathcal{U}$. 
For $d \in \mathbb{N}$, we define $\mathcal{S}_d(\vec W)$ as the set of kernels 
\[K:\Lambda_N\times\Lambda_N\to \mathbb C:\left(t,h;t',h'\right)\mapsto K\left(t,h;t',h'\right),\] 
for which there exist column vectors $\vec \rho_t = (\rho_{t,j})_{j=1}^d$ and $\vec \sigma_t = (\sigma_{t,j})_{j=1}^d$ of complex-valued functions such that $K$ admits the following double-contour integral form:
\begin{multline} \label{def: kernel}
        K\big(t,h;t',h'\big)= -1_{t>t'} \oint_{\Sigma_1}  z^{-h} \frac{W_{t'}(z)}{W_t(z)} z^{h'} \frac{dz}{2\pi iz}
         \\
         + \oint_{\Sigma_2} \oint_{\Sigma_1} u^{-h} \frac{1}{W_{t}(u)} \frac{v}{u-v} \vec\rho_t(u)^T \vec\sigma_{t'}(v)   W_{t'}(v) v^{h'} \frac{dv}{2\pi iv} \frac{du}{2\pi iu},
\end{multline}
where $\Sigma_1$ is the positively oriented unit circle and $\Sigma_2$ is a positively oriented simple closed contour outside the unit circle and inside $\mathcal{U}$, as depicted in Figure \ref{fig:contours}, such that for $t = 1, \dots, N$, the functions $\frac{1}{W_t} \vec\rho_{t}$ are analytic in a neighbourhood of $\Sigma_2$ and the functions $W_t\vec\sigma_{t}$ are analytic in a neighbourhood of $\Sigma_1$.

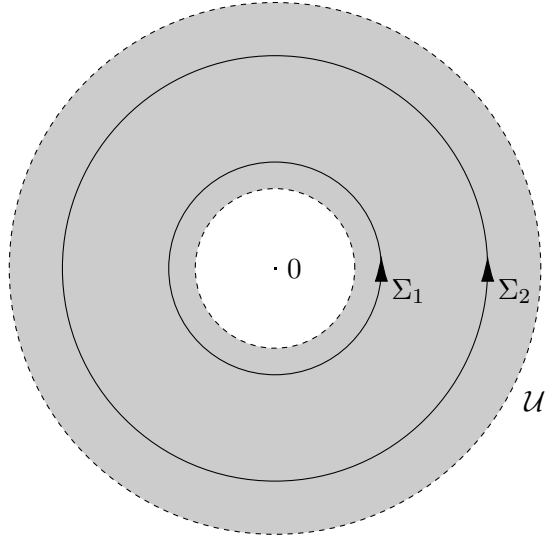
\begin{figure}
    \centering
    \begin{tikzpicture}
        \fill [opacity = 0.2, black] (0,0) circle[radius=100pt];
        \fill [white] (0,0) circle[radius=30pt];
        \draw circle [radius = 40pt];
        \coordinate [label=-45:\textcolor{black}{$\Sigma_1$}]  (1) at (40pt,0pt);
        \draw circle [radius = 80pt];
        \coordinate [label=-45:\textcolor{black}{$\Sigma_2$}]  (2) at (80pt,0pt);
        \draw [dash=on 2pt off 2pt phase 0pt] circle [radius = 30pt];
        \draw [dash=on 2pt off 2pt phase 0pt] circle [radius = 100pt];
        \draw [arrows = {-Latex[width=5pt, length=10pt]}] (40pt,5pt);
        \draw [arrows = {-Latex[width=5pt, length=10pt]}] (80pt,5pt);
        \node [fill=black,inner sep=0.5pt,label=0:$0$] at (0,0) {};
        \coordinate [label=0:\textcolor{black}{$\mathcal{U}$}] (3) at (90pt,-50pt);
    \end{tikzpicture}
    \caption{The contours of integration for a kernel belonging to the space $\mathcal{S}_d(\vec W)$. The annular neighbourhood $\mathcal{U}$ is highlighted in blue.}
    \label{fig:contours}
\end{figure}

\medskip

Correlation kernels of the form \eqref{def: kernel} appear in a variety of discrete multi-time determinantal point processes, the most important general class being the Schur processes \cite{okunkov2002, okunkov2003} which correspond to $d=1$, $\vec\sigma=\vec\rho=1$, for a large class of functions $W_t$. 
One of the simplest concrete examples of such a multi-time determinantal point process in the Schur class is the Krawtchouk ensemble introduced by Johansson in \cite{Johansson1999DiscreteOP, Johansson2000NonintersectingPR, Johansson2003TheAC}. Random point configurations in this ensemble are in one-to-one correspondence with weighted random domino tilings of the Aztec diamond, with weight $1$ assigned to horizontal dominoes and weight $a\in (0,1)$ assigned to vertical dominoes; see Section \ref{section:4}.
In this ensemble, the kernel is of the form \eqref{def: kernel} with $d=1$ and $\vec\rho=\vec\sigma=1$.
Moreover, writing $N=2M$ with $M\in\mathbb N$ and $t = 2s - \epsilon$
for $s\in\{1,\ldots, N\}$ and $\epsilon \in\{ 0, 1\}$, the functions $W_t$ are meromorphic functions in the complex plane given by
\begin{equation}
    W_t(z) = \frac{z^{2M-s+\epsilon - 1}}{(z-a)^{2M-s+\epsilon}(1+az)^{s}}.
\end{equation}
In more complicated models, the vectors $\vec \rho_t$ and $\vec \sigma_{t}$ may not have an explicit expression and may be characterized in terms of the solution of a Riemann-Hilbert (RH) problem. This is the case for lozenge tilings of hexagons, for which $d=2$, and for periodic tiling models \cite{Charlier2020DoublyPL, Duits2017TheTA}, where $d$ can be arbitrarily large.

\medskip

The general theory of determinantal point processes \cite{Soshnikov_2000, borodin2015dpp, JOHANSSON2006dpp} tells us that void probabilities of $\mathbb{P}$ are given by Fredholm determinants of $K$. To make this precise we need to interpret $K$ as the kernel of an integral operator (which we denote by the same symbol $K$) acting on $\ell^2(\Lambda_N)$:
\begin{equation}
    K[\sigma](t,h) = \sum_{t'=1}^N \sum_{h' \in \mathbb{Z}} K(t,h;t',h') \sigma(t',h'), \quad \sigma \in \ell^2(\Lambda_N).
\end{equation}
Given a subset $B \subset \Lambda_N$, the void (or gap) probability is the probability that a random point configuration $\mathcal X$ has no points in $B$. We denote $1_B$ for the projection operator on $B$ defined by
\begin{equation}
    1_B[\sigma](t,h)=\begin{cases}\sigma(t,h),& (t,h)\in B,\\
0,&(t,h)\in\Lambda_N\setminus B,\end{cases},\quad \sigma\in\ell^2(\Lambda_N).
\end{equation}
If $B$ is such that $1_BK$ defines a trace-class operator, then 
\begin{equation}\label{def:void}
\mathbb P \left(\mathcal{X} \cap B = \emptyset \right)=\det\left(1-1_BK\right)_{\ell^2\left(\Lambda_N\right)},
\end{equation}
where the right-hand side is the Fredholm determinant which can be evaluated as the Fredholm series
\begin{equation}\label{def:Fredholm}
\det\left(1-1_BK\right)_{\ell^2\left(\Lambda_N\right)}=1+\sum_{m=1}^\infty \frac{(-1)^m}{m!}\sum_{B^m}
\det\left(K\left(t_j,h_j;t_k,h_k\right)\right)_{j,k=1}^m.
\end{equation}
The second sum runs over all $m$-tuples of points $\left((t_1,h_1),\ldots, (t_m,h_m)\right)$ in $B^m$. 

\medskip

We consider finite subsets $B$ such that the operator $1_BK$ has finite rank and thus is trace-class. We conveniently parametrise $B$ as a union of clusters of consecutive points corresponding to the same time level. Given $t \in\{1, \dots, N\}$ and $a< b \in \Z$, we define the cluster
\begin{equation} \label{def: pregap}
    B_{t,a,b} := \{t\} \times \{a,a+1,\dots, b-1\} \subset \Lambda_N.
\end{equation}
Then, for $q \in \mathbb{N}$ and for vectors $\vec t = (t_1, \dots, t_q)$, $\vec a = (a_1, \dots, a_q)$, and $\vec b = (b_1, \dots, b_q)$ satisfying $N\geq t_1\geq \cdots\geq t_q\geq 1$ and $a_j<b_j$ for every $j=1,\ldots, q$, we define
\begin{equation} \label{def: gap}
    B_{\vec t, \vec a, \vec b} := \bigcup_{j=1}^q B_{t_j,a_j,b_j}.
\end{equation}
In other words, $B_{\vec t, \vec a, \vec b}$ contains $q$ clusters of points $\left\{(t_j,a_j), (t_j,a_j+1),\ldots, (t_j,b_j-1)\right\}$ each containing a group of consecutive lattice points at one time level, see Figure \ref{fig:lattice}.
We assume moreover that the clusters are disjoint and that 
there is at least one empty site between two clusters at the same time level (if $t_j=t_k$ for $j\neq k$, then $[a_j,b_j]$ and $[a_k,b_k]$ are disjoint).  
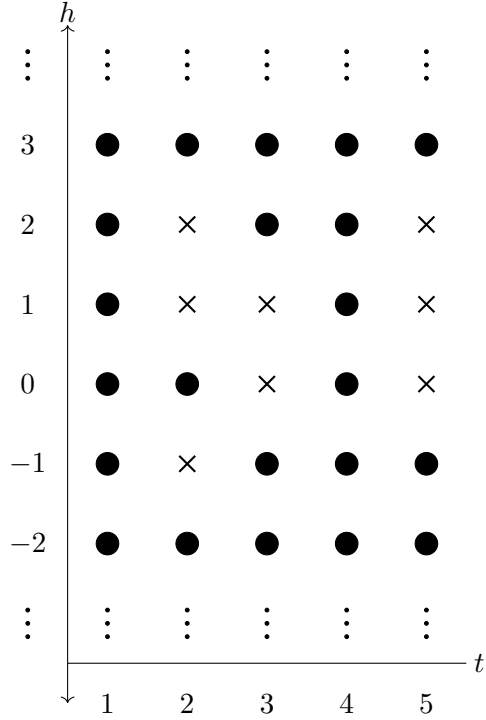
\begin{figure}
    \centering
    \begin{tikzpicture}
        \tikzset{
            dots/.pic={
                \node at (0pt,5pt) [circle, inner sep = 0.5pt, fill=black, draw] {};
                \node at (0pt,0pt) [circle, inner sep = 0.5pt, fill=black, draw] {};
                \node at (0pt,-5pt) [circle, inner sep = 0.5pt, fill=black, draw] {};
            }
        }

        \tikzset{
            cross/.pic={
                \draw (-3pt,-3pt) edge [thick] (3pt,3pt);
                \draw (-3pt,3pt) edge [thick] (3pt,-3pt);
            }
        }

        \node (1) at (-60pt,-120pt)     [black]   {$1$};
        \node (2) at (-30pt,-120pt)     [black]   {$2$};
        \node (3) at (0pt,-120pt)     [black]   {$3$};
        \node (4) at (30pt,-120pt)     [black]   {$4$};
        \node (5) at (60pt,-120pt)     [black]   {$5$};

        \pic at (-90pt, 120pt) {dots};
        \pic at (-60pt, 120pt) {dots};
        \pic at (-30pt, 120pt) {dots};
        \pic at (0pt, 120pt) {dots};
        \pic at (30pt, 120pt) {dots};
        \pic at (60pt, 120pt) {dots};

        \pic at (-90pt, -90pt) {dots};
        \pic at (-60pt, -90pt) {dots};
        \pic at (-30pt, -90pt) {dots};
        \pic at (0pt, -90pt) {dots};
        \pic at (30pt, -90pt) {dots};
        \pic at (60pt, -90pt) {dots};

        \node (3) at (-90pt,90pt)     [black]   {$3$};
        \node at (-60pt,90pt) [circle, inner sep = 3pt, fill=black, draw] {};
        \node at (-30pt,90pt) [circle, inner sep = 3pt, fill=black, draw] {};
        \node at (0pt,90pt) [circle, inner sep = 3pt, fill=black, draw] {};
        \node at (30pt,90pt) [circle, inner sep = 3pt, fill=black, draw] {};
        \node at (60pt,90pt) [circle, inner sep = 3pt, fill=black, draw] {};
        
        \node (2) at (-90pt,60pt)     [black]   {$2$};
        \node at (-60pt,60pt) [circle, inner sep = 3pt, fill=black, draw] {};
        \pic at (-30pt,60pt) {cross};
        \node at (0pt,60pt) [circle, inner sep = 3pt, fill=black, draw] {};
        \node at (30pt,60pt) [circle, inner sep = 3pt, fill=black, draw] {};
        \pic at (60pt,60pt) {cross};

        \node (1) at (-90pt,30pt)     [black]   {$1$};
        \node at (-60pt,30pt) [circle, inner sep = 3pt, fill=black, draw] {};
        \pic at (-30pt,30pt) {cross} ;
        \pic at (0pt,30pt) {cross} ;
        \node at (30pt,30pt) [circle, inner sep = 3pt, fill=black, draw] {};
        \pic at (60pt,30pt) {cross} ;

        \node (0) at (-90pt,0pt)     [black]   {$0$};
        \node at (-60pt,0pt) [circle, inner sep = 3pt, fill=black, draw] {};
        \node at (-30pt,0pt) [circle, inner sep = 3pt, fill=black, draw] {};
        \pic at (0pt,0pt) {cross} ;
        \node at (30pt,0pt) [circle, inner sep = 3pt, fill=black, draw] {};
        \pic at (60pt,0pt) {cross} ;

        \node (-1) at (-90pt,-30pt)     [black]   {$-1$};
        \node at (-60pt,-30pt) [circle, inner sep = 3pt, fill=black, draw] {};
        \pic at (-30pt,-30pt) {cross} ;
        \node at (0pt,-30pt) [circle, inner sep = 3pt, fill=black, draw] {};
        \node at (30pt,-30pt) [circle, inner sep = 3pt, fill=black, draw] {};
        \node at (60pt,-30pt) [circle, inner sep = 3pt, fill=black, draw] {};

        \node (-2) at (-90pt,-60pt)     [black]   {$-2$};
        \node at (-60pt,-60pt) [circle, inner sep = 3pt, fill=black, draw] {};
        \node at (-30pt,-60pt) [circle, inner sep = 3pt, fill=black, draw] {};
        \node at (0pt,-60pt) [circle, inner sep = 3pt, fill=black, draw] {};
        \node at (30pt,-60pt) [circle, inner sep = 3pt, fill=black, draw] {};
        \node at (60pt,-60pt) [circle, inner sep = 3pt, fill=black, draw] {};

        \draw [<->] (-75pt,-120pt) -- (-75pt,135pt);
        \draw [-] (-75pt,-105pt) -- (75pt,-105pt);
        \node (h) at (-75pt,140pt) [black] {$h$};
        \node (t) at (80pt,-105pt) [black] {$t$};
    \end{tikzpicture}
    
    \caption{An illustration of $\Lambda_5$ with the set $B_{\vec t, \vec a, \vec b}$ marked with crosses, where $\vec t = (5,3,2,2)$, $\vec a=(0, 0, 1, -1)$, and $\vec b = (3, 2, 3, 0)$.}
    \label{fig:lattice}
\end{figure}

\medskip

If the void probability $\mathbb{P}\left(\mathcal{X} \cap B_{\vec t, \vec a, \vec b}=\emptyset\right)$ is non-zero, then we can form a new point process $\mathbb{P}_{\vec t, \vec a, \vec b}$ on $\Lambda_N \setminus B_{\vec t, \vec a, \vec b}$ by conditioning on the event $\mathcal{X} \cap B_{\vec t, \vec a, \vec b} = \emptyset$. 
It is known that the point process $\mathbb{P}_{\vec t, \vec a, \vec b}$ is determinantal \cite[Theorem 1.1 with $\theta=1_B$]{{Claeys2021DeterminantalPP}} and that its correlation kernel $K_{\vec t, \vec a, \vec b}\left((t,h), (t',h')\right)$ for $(t,h), (t',h')\in\Lambda_N\setminus B_{\vec t, \vec a, \vec b}$ 
is equal to the integral kernel of the  $\ell^2(\Lambda_N)$-operator
\begin{equation} \label{def: K1}
    K_{\vec t, \vec a, \vec b} := K\left(1-1_{B_{\vec t, \vec a, \vec b}}K\right)^{-1}.
\end{equation}
We emphasize that the kernel $K_{\vec t, \vec a, \vec b}$ of this operator is defined on $\Lambda_N \times \Lambda_N$, but that it is the restriction of $K_{\vec t, \vec a, \vec b}$ to $\left(\Lambda_N \setminus B_{\vec t, \vec a, \vec b}\right)\times\left(\Lambda_N \setminus B_{\vec t, \vec a, \vec b}\right)$
which defines the determinantal point process $\mathbb{P}_{\vec t, \vec a, \vec b}$.

\medskip

In extension, we consider a random thinning of $\mathbb{P}$; for fixed $\lambda \in (0,1)$, given a configuration $\mathcal{X}$, we keep each point $x \in \mathcal{X}$ with probability $\lambda$, otherwise we remove it. In doing so, we produce a new thinned point process $\mathbb{P}^{\lambda}$. Void probabilities of $\mathbb{P}^{\lambda}$ are again given by Fredholm determinants of $K$; for $B \subset \Lambda$, we have
\begin{align}
    \mathbb{P}^{\lambda}\left(\mathcal{X} \cap B\right) &= \det\left(1-\lambda 1_BK\right)_{\ell^2\left(\Lambda_N\right)} \\
    &=1+\sum_{m=1}^\infty \frac{(-\lambda)^m}{m!}\sum_{B^m}
\det\left(K\left(t_j,h_j;t_k,h_k\right)\right)_{j,k=1}^m.
\end{align}
Now if the void probability $\mathbb{P}^{\lambda} \left(\mathcal{X} \cap B_{\vec t, \vec a, \vec b}=\emptyset\right)$ is non-zero, then we can form a new point process $\mathbb{P}_{\vec t, \vec a, \vec b,\lambda}$ on $\Lambda_N$ by conditioning on the event that the thinned process $\mathbb{P}^{\lambda}$ has no points in $B_{\vec t, \vec a, \vec b}$. By \cite[Theorem 1.1 with $\theta=\lambda 1_B$]{Claeys2021DeterminantalPP}, $\mathbb{P}_{\vec t, \vec a, \vec b,\lambda}$ defines a determinantal point process; if we define, for $\lambda \in (0,1)$,
\begin{equation} \label{def: Ktransform}
    K_{\vec t, \vec a, \vec b, \lambda} := K \left(1 - \lambda 1_{B_{\vec t, \vec a, \vec b}}K\right)^{-1},
\end{equation}
then the correlation kernel for $\mathbb{P}_{\vec t, \vec a, \vec b,\lambda}$ is
\[\left(1 - \lambda 1_{B_{\vec t, \vec a, \vec b}}\right)K_{\vec t, \vec a, \vec b, \lambda}.\]
We can interpret $\mathbb{P}_{\vec t, \vec a, \vec b,\lambda}$ as the point process $\mathbb P$ under the influence of a {\em soft} reduction of the domain, in which smaller weights are assigned to configurations with points in $B_{\vec t, \vec a, \vec b}$,
while $\mathbb P_{\vec t,\vec a,\vec b}$ is the point process $\mathbb P$ under the influence of a {\em hard} domain reduction, in which zero weight is assigned to configurations with points in $B_{\vec t, \vec a, \vec b}$.

\medskip

Many kernels of the form \eqref{def: kernel} do not define a determinantal point process; this is the case, for instance, if the resulting correlation functions are not positive. However in some such cases, the kernel may still define a signed or complex-valued determinantal measure, see e.g. \cite{cafasso2026biorthogonal}. Our main results, which we will describe next, do not rely on the existence of an underlying determinantal point process and only require the operator $1-1_{B_{\vec t, \vec a, \vec b}}K$ to be invertible. However, there is only a clear probabilistic meaning in terms of a conditional point process if $K$ defines a (positive) determinantal point process.

\begin{theorem}\label{thm:transformation}
Let $K$ belong to the space $\mathcal S_d(\vec W)$. Let $B_{\vec t, \vec a, \vec b}$ be a finite subset of $\Lambda_N$ given by \eqref{def: gap} for some $q\in\mathbb N$. If $1-\lambda 1_{B_{\vec t, \vec a, \vec b}}K$ is invertible as an $\ell^2(\Lambda_N)$-operator, then the kernel $K_{\vec t, \vec a, \vec b, \lambda}$ of the operator \eqref{def: Ktransform} belongs to the space $\mathcal S_{d+2q}(\vec W)$.
\end{theorem}

Next we want to describe the transformation $K\in\mathcal S_d(\vec W) \mapsto K_{\vec t,\vec a,\vec b,\lambda} \in S_{d+2q}(\vec W)$ in more detail. We will do this by explaining the effect of the transformation on the functions $\vec\rho_t$ and $\vec\sigma_t$.
The crucial feature of the expression \eqref{def: kernel} is that the function
\begin{equation}\label{def:kernel almost integrable}\frac{v}{u-v} \vec \rho_t(u)^T \vec \sigma_{t'}(v),\end{equation} 
appearing in the double-contour integral is of integrable type in the sense of Its, Izergin, Korepin, and Slavnov \cite{IIKS}. We recall from \cite{IIKS, HarnadIts} that a $p\times p$ matrix-valued kernel $M:\Sigma\times\Sigma\to\mathbb C^{p\times p}$, where $\Sigma$ is a union of smooth curves in the complex plane without intersections, is said to be of $(n\times p)$-integrable type if it can be written in the form 
\begin{equation} \label{def: integrable kernel}
    M(z,w) = \frac{w}{z-w}F(z)^TG(w)\qquad \mbox{with }F(z)^TG(z)\equiv0,
\end{equation}
where $F$ and $G$ are smooth $n\times p$ matrix-valued functions defined on $\Sigma$. Our kernel \eqref{def:kernel almost integrable} is indeed integrable, with $p=1$ and $n=d$: it suffices to set $\Sigma=\Sigma_1\cup\Sigma_2$ and
\begin{equation}
    F(z)=1_{z\in\Sigma_2}\vec\rho_t(z),\qquad G(w)=1_{w\in\Sigma_1}\vec\sigma_{t'}(w).
\end{equation}
Note that we view \eqref{def:kernel almost integrable} as the kernel of an integral operator acting on the space $L^2\left(\Sigma_1 \cup \Sigma_2;\frac{dz}{2\pi iz}\right)$.

\medskip

The results in \cite{IIKS, DIZ} imply that for an integrable kernel $M$ given by \eqref{def: integrable kernel}, the kernel of the resolvent operator
\begin{equation} \label{def: resolvent}
    R := M\left(1-M\right)^{-1},
\end{equation}
if it exists, is again integrable and can be expressed 
in terms of an $n \times n$ Riemann-Hilbert (RH) problem. Specialised to our setting where $\Sigma = \Sigma_1 \cup \Sigma_2$, we write $F,G$ in the form
\begin{align}
    F(z) &= 1_{z \in \Sigma_1}F_1(z) + 1_{z \in \Sigma_2}F_2(z), \label{def: F12} \\
    G(w) &= 1_{w \in \Sigma_1}G_1(w) + 1_{w \in \Sigma_2}G_2(w). \label{def: G12}
\end{align}
The precise RH problem connected to the resolvent operator \eqref{def: resolvent} is the following.
\begin{problem}\label{RHPUpsilonIIKS} 
    \begin{enumerate}
        \item $\Upsilon: \mathbb{C} \setminus (\Sigma_1 \cup \Sigma_2) \rightarrow \mathbb{C}^{n \times n}$ is analytic.
        \item 
        For $z \in \Sigma_1 \cup \Sigma_2$, $\Upsilon$ has continuous boundary values, $\Upsilon_+$ and $\Upsilon_-$, and they satisfy the jump relations
        \begin{align}
            \Upsilon_+(z) &= \Upsilon_-(z)\left(I_n + G_1(z)F_1(z)^T\right), \quad \text{for $z \in \Sigma_1$,} \\
            \Upsilon_+(z) &= \Upsilon_-(z)\left(I_n + G_2(z)F_2(z)^T\right), \quad \text{for $z \in \Sigma_2$.}
        \end{align}
        \item $\Upsilon(z) = I_{n} + \mathcal{O}(1/z)$ as $z \rightarrow \infty$.
    \end{enumerate}
    \end{problem}
A unique solution $\Upsilon$ of RH problem \ref{RHPUpsilonIIKS} exists if and only if $1-M$ is invertible. We have that, for $z,w \in \Sigma_1 \cup \Sigma_2$,
\begin{equation} \label{def: Rintegrable}
    R(z,w) = \frac{w}{z-w} \tilde{F}(z)^T \tilde{G}(w),
\end{equation}
where
\begin{align}
    \tilde F(z) &:= 1_{z \in \Sigma_1} \Upsilon_\pm(z) ^{-T} F_1(z) + 1_{z \in \Sigma_2} \Upsilon_\pm(z) ^{-T} F_2(z), \label{def: Ftilde} \\
    \tilde G(w) &:= 1_{w \in \Sigma_1} \Upsilon_\pm(w) G_1(w) + 1_{w \in \Sigma_2} \Upsilon_\pm(w) G_2(w). \label{def: Gtilde}
\end{align}
It is readily checked that \eqref{def: Ftilde}--\eqref{def: Gtilde} are independent of the choice of $+$ or $-$ boundary values: indeed, we have that for $i,j = 1,2$ and $z \in \Sigma_i$ and $w \in \Sigma_j$,
\begin{align}
    \Upsilon_+^{-T}(z)F_i(z)-\Upsilon_-^{-T}(w)F_i(z) &= -\Upsilon_-^{-T}(w)F_i(z)G_i(z)^TF_i(z) = 0, \label{def: Upsilonjump1} \\
    \Upsilon_+(w)G_j(w) - \Upsilon_-(w)G_j(w) &= \Upsilon_-(w)G_j(w)F_j(w)^TG_j(w) = 0. \label{def: Upsilonjump2}
\end{align}
The theory of integrable operators and their resolvents was developed in \cite{IIKS, DIZ}. A detailed description of the theory for matrix-valued kernels is given in \cite[Section 1b]{HarnadIts}. 

\medskip

To present our next result, we express kernels belonging to $\mathcal{S}_d(\vec W)$ in a slightly different, at first sight more complicated than \eqref{def: kernel}, form, by setting, for $u \in \Sigma_2$ and $v \in \Sigma_1$, (notice the analogy with $\tilde F,\tilde G$ in \eqref{def: Ftilde}--\eqref{def: Gtilde})
\begin{equation}\label{def:phipsi}
\vec\rho_t(u)=Y_\pm(u)^{-T}\vec\phi_t(u),\qquad \vec\sigma_{t'}(v)=Y_\pm(v)\vec\psi_{t'}(v),
\end{equation}
where $Y$ is the solution of a $d \times d$ RH problem satisfying jump relations on the contours $\Sigma_1$ and $\Sigma_2$, and where $\vec\phi_t = (\phi_{t,j})_{j=1}^d$ and $\vec\psi_{t} = (\psi_{t,j})_{j=1}^d$ are column vectors of complex-valued functions. In our definition of $\mathcal{S}_d(\vec W)$, we imposed that for $t = 1 \dots, N$, the functions $\frac{1}{W_t}\vec\rho_{t}$ are analytic in a neighbourhood of $\Sigma_2$ and the functions $W_t\vec\sigma_{t}$ are analytic in a neighbourhood of $\Sigma_1$. For this reason, we impose here that the functions
$\frac{1}{W_t}\vec\phi_{t}$ are analytic in a neighbourhood of $\Sigma_2$ and that the functions $W_t\vec\psi_{t}$ are analytic in a neighbourhood of $\Sigma_1$, and moreover that
for $u \in \Sigma_2$ and $v \in \Sigma_1$,
\begin{equation} \label{expr: no jump}
    Y_+^{-T}(u) \vec \phi_t(u) = Y_-^{-T}(u) \vec \phi_t(u), \quad Y_+(v) \vec \psi_t(v) = Y_-(v) \vec \psi_t(v).
\end{equation}
Then it follows easily that 
$\frac{1}{W_t}\vec\rho_{t}$ and $W_t\vec\sigma_{t}$ can be extended analytically to neighbourhoods of $\Sigma_2$ and $\Sigma_1$.

For any kernel belonging to $\mathcal{S}_d(\vec W)$, there always exists a canonical choice consisting of setting
\begin{equation} \label{def:canonical}
Y\equiv I_d,\qquad \vec\phi_t=\vec\rho_t,\qquad \vec\psi_{t'}=\vec\sigma_{t'},
\end{equation} 
but in some cases there may be more convenient choices of $Y$, as we will see later (see in particular Remark \ref{remark:choiceY}). Thus, we can write any kernel $K$ belonging to $\mathcal{S}_d(\vec W)$ in the form
\begin{multline} \label{def: kernelRHP}
        K\big(t,h;t',h'\big)= -1_{t>t'} \oint_{\Sigma_1}  z^{-h} \frac{W_{t'}(z)}{W_t(z)} z^{h'} \frac{dz}{2\pi iz}
         \\
         + \oint_{\Sigma_2} \oint_{\Sigma_1} u^{-h} \frac{1}{W_{t}(u)} \frac{v}{u-v} \vec\phi_t(u)^T Y_\pm(u)^{-1}Y_\pm(v) \vec\psi_{t'}(v)   W_{t'}(v) v^{h'} \frac{dv}{2\pi iv} \frac{du}{2\pi iu},
\end{multline}
where $Y$ solves a RH problem of the following form.
\begin{problem}\label{RHPYmodel}\begin{enumerate}
        \item $Y: \mathbb{C} \setminus (\Sigma_1 \cup \Sigma_2) \rightarrow \mathbb{C}^{d \times d}$ is analytic.
        \item 
        For $z \in \Sigma_1 \cup \Sigma_2$, $Y$ has continuous boundary values, $Y_+$ and $Y_-$, and they satisfy jump relations of the form
        \begin{align}
            Y_+(z) &= Y_-(z)J^{(1)}(z), \quad \text{for $z \in \Sigma_1$,} \\
            Y_+(z) &= Y_-(z)J^{(2)}(z), \quad \text{for $z \in \Sigma_2$,}
        \end{align}
        where $J^{(1)}$ and $J^{(2)}$ satisfy the relations
        \begin{align}\label{jump rel1}
    J^{(2)}(u)^{-T} \vec \phi_t(u) &=  \vec \phi_t(u), && u\in\Sigma_2, \\
\label{jump rel2}    J^{(1)}(v) \vec \psi_{t'}(v) &=  \vec \psi_{t'}(v),&& v\in\Sigma_1.
\end{align}
        \item $Y(z) = \left(I_{d} + \mathcal{O}(1/z)\right)Y^{(\infty)}(z)$ as $z \rightarrow \infty$,
        for some $d\times d$ matrix-valued function $Y^{(\infty)}(z)$ which is defined for sufficiently large $z$ and satisfies $\det Y^{(\infty)}\equiv 1$.
    \end{enumerate}
    \end{problem}
Note that the canonical choice \eqref{def:canonical} corresponds to setting $J^{(1)} = J^{(2)} = Y^{(\infty)} = I_d$. One verifies that the relations \eqref{jump rel1}--\eqref{jump rel2} imply that \eqref{expr: no jump} is satisfied.

\medskip

Under the map $K \mapsto K_{\vec t, \vec a, \vec b,\lambda}$, only the expression $\vec\phi_t(u)^T Y_\pm(u)^{-1}Y_\pm(v) \vec\psi_{t'}(v)$ in \eqref{def: kernelRHP} changes. We describe this transformation with four maps
\begin{align}
    &\vec \psi_t(v) \mapsto \vec \psi_t(v;\vec t, \vec a, \vec b, \lambda),  & \vec \phi_t(u) \mapsto \vec \phi_t(u;\vec t, \vec a, \vec b, \lambda), \\
    & J^{(1)}(v) \mapsto J^{(1)}(v;\vec t, \vec a, \vec b, \lambda),  & J^{(2)}(u) \mapsto J^{(2)}(u;\vec t, \vec a, \vec b, \lambda),
\end{align}
defined as follows:
We first define,
for $j = 1, \dots, q$, two column vectors
\begin{equation} \label{def: fg}
    \vec f_j(z) := \lambda \begin{pmatrix}
        z^{b_j}  \\
        -z^{a_j}
    \end{pmatrix}, \quad \vec g_j(z) := \begin{pmatrix}
        z^{-b_j} \\ z^{-a_j}
    \end{pmatrix}.
\end{equation}
Then we define the transformed jump matrices by
\begin{multline} \label{def: J1}
    J^{(1)}(z;\vec t, \vec a, \vec b, \lambda) := \\
    \left[\begin{array}{c|cccc}
        J^{(1)}(z) & \vec \psi_{t_1}(z)W_{t_1}(z) \vec f_{1}(z)^T & \vec \psi_{t_2}(z)W_{t_2}(z)\vec f_2(z)^T & \cdots & \vec \psi_{t_q}(z)W_{t_q}(z)\vec f_q(z)^T \\ \hline
        0_{2 \times d} & I_2 & -1_{t_1>t_2} \vec g_1(z) \frac{W_{t_2}(z)}{W_{t_1}(z)}\vec f_2(z)^T & \cdots & -1_{t_1>t_q} \vec g_1(z) \frac{W_{t_q}(z)}{W_{t_1}(z)}\vec f_q(z)^T \\
         0_{2 \times d} &  0_{2 \times 2} & I_2 & \cdots & -1_{t_2>t_q} \vec g_2(z) \frac{W_{t_q}(z)}{W_{t_2}(z)}\vec f_q(z)^T \\
        \vdots & \vdots & \vdots & \ddots & \vdots \\
        0_{2 \times d} &  0 _{2\times 2} & 0_{2 \times 2} & \cdots & I_2
    \end{array}\right],
\end{multline}
where we notice that the indicators $1_{t_j>t_k}$ vanish only if $t_j=t_k$, and
\begin{equation} \label{def: J2}
    J^{(2)}(z;\vec t, \vec a, \vec b, \lambda) :=
    \left[\begin{array}{c|cccc}
        J^{(2)}(z) &  0_{d \times 2} & 0_{d \times 2} & \cdots &  0_{d \times 2} \\ \hline
        \vec g_1(z) \frac{1}{W_{t_1}(z)} \vec \phi_{t_1}(z)^T & I_2 &  0_{2\times 2} & \cdots &  0_{2\times 2} \\
        \vec g_2(z) \frac{1}{W_{t_2}(z)} \vec \phi_{t_2}(z)^T &  0_{2 \times 2} & I_2 & \cdots &  0_{2\times 2} \\
        \vdots & \vdots & \vdots & \ddots & \vdots \\
        \vec g_q(z) \frac{1}{W_{t_q}(z)} \vec \phi_{t_q}(z)^T &  0_{2\times 2} &  0_{2\times 2} & \cdots & I_2 \\
    \end{array}\right].
\end{equation}
Finally, we define the transformed column vectors by
\begin{equation} \label{def: phipsi}
    \vec \phi_t(u;\vec t, \vec a, \vec b, \lambda) := \begin{bmatrix}
        \vec \phi_t(u) \\ -1_{t>t_1} \vec f_1(u) W_{t_1}(u) \\ \vdots \\ -1_{t>t_q} \vec f_q(u) W_{t_q}(u)
    \end{bmatrix}, \quad \vec \psi_{t'}(v;\vec t, \vec a, \vec b, \lambda) := \begin{bmatrix}
        \vec \psi_{t'}(v) \\ -1_{t_1>t'} \vec g_1(v) \frac{1}{W_{t_1}(v)} \\
        \vdots \\
        -1_{t_q>t'} \vec g_q(v) \frac{1}{W_{t_q}(v)}
    \end{bmatrix}.
\end{equation}
These transformations simply append entries to the vectors and matrices depending on the choice of the set $B$ such that row and column dimensions increase from $d$ to $d+2q$. The new jump matrices $J^{(1)}$ and $J^{(2)}$ define a new RH problem for $Y(\cdot;\vec t, \vec a, \vec b, \lambda)$ of size $(d+2q)\times (d+2q)$.

\begin{problem} \label{prbl: main}
    \begin{enumerate}
            \item $Y(\cdot;\vec t, \vec a, \vec b, \lambda): \mathbb{C} \setminus (\Sigma_1 \cup \Sigma_2) \rightarrow \mathbb{C}^{(d+2q) \times (d+2q)}$ is analytic.
            \item 
            For $z \in \Sigma_1 \cup \Sigma_2$, $Y(\cdot;\vec t, \vec a, \vec b, \lambda)$ has continuous boundary values, $Y_+(z;\vec t, \vec a, \vec b, \lambda)$ and $Y_-(z;\vec t, \vec a, \vec b, \lambda)$, and they satisfy the jump relations
            \begin{align}
            Y_+(z;\vec t, \vec a, \vec b, \lambda)&=Y_-(z;\vec t, \vec a, \vec b, \lambda)J^{(1)}(z;\vec t, \vec a, \vec b, \lambda), \qquad & \text{for $z\in\Sigma_1$},\\
            Y_+(z;\vec t, \vec a, \vec b, \lambda)&=Y_-(z;\vec t, \vec a, \vec b, \lambda)J^{(2)}(z;\vec t, \vec a, \vec b, \lambda), \qquad & \text{for $z\in\Sigma_2$}.
            \end{align}
            \item $Y(z;\vec t, \vec a, \vec b, \lambda) = \left(I_{d+2q} + \mathcal{O}(1/z)\right) \begin{bmatrix}
                Y^{(\infty)}(z) & \vec0_{d \times 2q} \\
                \vec 0_{2q \times d} & I_{2q}
            \end{bmatrix}$ as $z \rightarrow \infty$.
        \end{enumerate}
\end{problem}

We are now ready to state our general result describing in full detail the transformation implied by Theorem \ref{thm:transformation}.
\begin{theorem} \label{thm: main}
    Let K belong to the space $\mathcal{S}_d(\vec W)$, given in the form \eqref{def: kernelRHP}.  Let $B_{\vec t, \vec a, \vec b}$ be a finite subset of $\Lambda_N$, given by \eqref{def: gap} for some $q\in\mathbb N$. If $1-\lambda 1_{B_{\vec t, \vec a, \vec b}}K$ is invertible as an $\ell^2(\Lambda_N)$-operator, then the operator $K_{\vec t, \vec a, \vec b, \lambda}$ given by \eqref{def: Ktransform} has integral kernel given by 
    \begin{multline}\label{def:kernelresult}
    K_{\vec t, \vec a, \vec b, \lambda} (t,h;t',h') = - 1_{t>t'} \oint_{\Sigma_1}z^{-h} \frac{W_{t'}(z)}{W_{t}(z)} z^{h'} \frac{dz}{2\pi iz} \\
    +\oint_{\Sigma_2}\oint_{\Sigma_1} u^{-h} \frac{1}{W_{t}(u)} \frac{v}{u-v} \vec \phi_t(u;\vec t, \vec a, \vec b, \lambda)^T Y_\pm(u;\vec t, \vec a, \vec b, \lambda)^{-1} \\
    \times Y_\pm(v;\vec t, \vec a, \vec b, \lambda) \vec \psi_{t'}(v;\vec t, \vec a, \vec b, \lambda) W_{t'}(v) v^{h'} \frac{dv}{2\pi iv} \frac{du}{2\pi iu},
\end{multline}
where $\vec \phi_t(u;\vec t, \vec a, \vec b, \lambda)$ and $\vec \psi_{t'}(v;\vec t, \vec a, \vec b, \lambda)$ are given by \eqref{def: phipsi} and where $Y(\cdot;\vec t, \vec a, \vec b, \lambda)$ is the unique solution of RH problem \ref{prbl: main} with jump matrices $J^{(1)}$ and $J^{(2)}$ given by \eqref{def: J1}--\eqref{def: J2}.
\end{theorem}

\begin{remark}
    Theorem \ref{thm:transformation} follows directly from Theorem \ref{thm: main}. Indeed, it is left to verify that $K_{\vec t, \vec a, \vec b,\lambda}$ belongs to the space $\mathcal{S}_{d+2q}(\vec W)$. This follows from observing that the new functions $\frac{1}{W_t}\vec \phi_t(\cdot; \vec t, \vec a, \vec b, \lambda)$ and $W_t\vec \psi_t(\cdot; \vec t, \vec a, \vec b, \lambda)$  remain analytic in neighbourhoods of $\Sigma_2$ and $\Sigma_1$ respectively and furthermore, for $u \in \Sigma_2$ and $v \in \Sigma_1$,
    \begin{align}
        J^{(2)}(u; \vec t, \vec a, \vec b, \lambda)^{-T} \vec \phi_t(u; \vec t, \vec a, \vec b, \lambda) &= \phi_t(u; \vec t, \vec a, \vec b, \lambda), \\
        J^{(1)}(v; \vec t, \vec a, \vec b, \lambda) \vec \psi_t(v; \vec t, \vec a, \vec b, \lambda) &= \psi_t(v; \vec t, \vec a, \vec b, \lambda),
    \end{align}
    such that $\frac{1}{W_t}\vec \rho_t(\cdot; \vec t, \vec a, \vec b, \lambda)$ and $W_t\vec \sigma_t(\cdot; \vec t, \vec a, \vec b, \lambda)$ remain analytic in neighbourhoods of $\Sigma_2$ and $\Sigma_1$ as well.
\end{remark}

\begin{remark}
We recall that if $K$ is the correlation kernel of a determinantal point process on $\Lambda_N$, then, for $\lambda \in (0,1)$, the kernel $\left(1 - \lambda 1_{B_{\vec t, \vec a, \vec b}} \right) K_{\vec t,\vec a,\vec b,\lambda}$ is the correlation kernel of the conditional point process $\mathbb P_{\vec t,\vec a,\vec b,\lambda}$, whereas $K_{\vec t,\vec a,\vec b}$ restricted to $\left( \Lambda \setminus B_{\vec t, \vec a, \vec b} \right) \times \left( \Lambda \setminus B_{\vec t, \vec a, \vec b} \right)$ is the correlation kernel of the conditional point process $\mathbb P_{\vec t,\vec a,\vec b}$.
\end{remark}

We can also use the solution of RH problem \ref{prbl: main} to obtain an exact expression of the Fredholm determinant $\det \left(1 - \lambda 1_{B_{\vec t, \vec a, \vec b}}K\right)_{\ell^2(\Lambda_N)}$. To describe this we first introduce a certain raising and lowering of the subset $B_{\vec t, \vec a, \vec b}$. Firstly, for each cluster $B_{t,a,b}$ given by \eqref{def: pregap}, we define the raised cluster by
\begin{equation} \label{def: pregap+}
    B_{t,a,b}^+ := B_{t,a,b+1},
\end{equation}
and the lowered cluster by
\begin{equation} \label{def: pregap-}
    B_{t,a,b}^- := B_{t,a+1,b}.
\end{equation}
For $B_{\vec t, \vec a, \vec b}$ given by \eqref{def: gap}, we define, for every $j = 1, \dots, q$, the $j$-th raised subset by
\begin{equation} \label{def: gap+}
    B_{\vec t, \vec a, \vec b}^{j+} := B_{t_j,a_j,b_j}^+ \cup B_{\vec t_{\setminus j}, \vec a_{\setminus j}, \vec b_{\setminus j}},
\end{equation}
and the $j$-th lowered subset by
\begin{equation} \label{def: gap-}
    B_{\vec t, \vec a, \vec b}^{j-} := B_{t_j,a_j,b_j}^- \cup B_{\vec t_{\setminus j}, \vec a_{\setminus j}, \vec b_{\setminus j}},
\end{equation}
where $\vec t_{\setminus j}$ is the vector $\vec t$ without its $j$-th component and likewise for $\vec a_{\setminus j}$ and $\vec b_{\setminus j}$. 

In other words, to obtain the $j$-th raised or lowered subset of $B_{\vec t, \vec a, \vec b}$, we either raise or lower the $j$-th cluster $B_{t_j,a_j,b_j}$, see Figure \ref{fig:latticepm}.

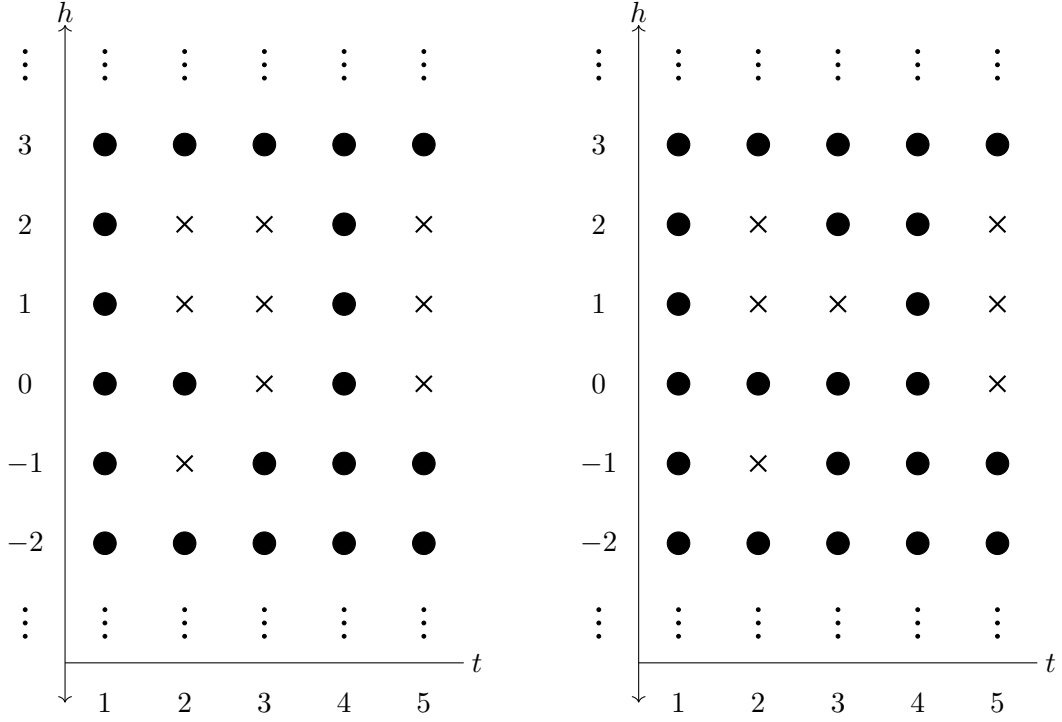
\begin{figure}
    \begin{subfigure}{0.5\textwidth}
    \begin{center}
    \begin{tikzpicture}
        \tikzset{
            dots/.pic={
                \node at (0pt,5pt) [circle, inner sep = 0.5pt, fill=black, draw] {};
                \node at (0pt,0pt) [circle, inner sep = 0.5pt, fill=black, draw] {};
                \node at (0pt,-5pt) [circle, inner sep = 0.5pt, fill=black, draw] {};
            }
        }

        \tikzset{
            cross/.pic={
                \draw (-3pt,-3pt) edge [thick] (3pt,3pt);
                \draw (-3pt,3pt) edge [thick] (3pt,-3pt);
            }
        }

        \node (1) at (-60pt,-120pt)     [black]   {$1$};
        \node (2) at (-30pt,-120pt)     [black]   {$2$};
        \node (3) at (0pt,-120pt)     [black]   {$3$};
        \node (4) at (30pt,-120pt)     [black]   {$4$};
        \node (5) at (60pt,-120pt)     [black]   {$5$};

        \pic at (-90pt, 120pt) {dots};
        \pic at (-60pt, 120pt) {dots};
        \pic at (-30pt, 120pt) {dots};
        \pic at (0pt, 120pt) {dots};
        \pic at (30pt, 120pt) {dots};
        \pic at (60pt, 120pt) {dots};

        \pic at (-90pt, -90pt) {dots};
        \pic at (-60pt, -90pt) {dots};
        \pic at (-30pt, -90pt) {dots};
        \pic at (0pt, -90pt) {dots};
        \pic at (30pt, -90pt) {dots};
        \pic at (60pt, -90pt) {dots};

        \node (3) at (-90pt,90pt)     [black]   {$3$};
        \node at (-60pt,90pt) [circle, inner sep = 3pt, fill=black, draw] {};
        \node at (-30pt,90pt) [circle, inner sep = 3pt, fill=black, draw] {};
        \node at (0pt,90pt) [circle, inner sep = 3pt, fill=black, draw] {};
        \node at (30pt,90pt) [circle, inner sep = 3pt, fill=black, draw] {};
        \node at (60pt,90pt) [circle, inner sep = 3pt, fill=black, draw] {};
        
        \node (2) at (-90pt,60pt)     [black]   {$2$};
        \node at (-60pt,60pt) [circle, inner sep = 3pt, fill=black, draw] {};
        \pic at (-30pt,60pt) {cross};
        \pic at (0pt,60pt) {cross};
        \node at (30pt,60pt) [circle, inner sep = 3pt, fill=black, draw] {};
        \pic at (60pt,60pt) {cross};

        \node (1) at (-90pt,30pt)     [black]   {$1$};
        \node at (-60pt,30pt) [circle, inner sep = 3pt, fill=black, draw] {};
        \pic at (-30pt,30pt) {cross} ;
        \pic at (0pt,30pt) {cross} ;
        \node at (30pt,30pt) [circle, inner sep = 3pt, fill=black, draw] {};
        \pic at (60pt,30pt) {cross} ;

        \node (0) at (-90pt,0pt)     [black]   {$0$};
        \node at (-60pt,0pt) [circle, inner sep = 3pt, fill=black, draw] {};
        \node at (-30pt,0pt) [circle, inner sep = 3pt, fill=black, draw] {};
        \pic at (0pt,0pt) {cross} ;
        \node at (30pt,0pt) [circle, inner sep = 3pt, fill=black, draw] {};
        \pic at (60pt,0pt) {cross} ;

        \node (-1) at (-90pt,-30pt)     [black]   {$-1$};
        \node at (-60pt,-30pt) [circle, inner sep = 3pt, fill=black, draw] {};
        \pic at (-30pt,-30pt) {cross} ;
        \node at (0pt,-30pt) [circle, inner sep = 3pt, fill=black, draw] {};
        \node at (30pt,-30pt) [circle, inner sep = 3pt, fill=black, draw] {};
        \node at (60pt,-30pt) [circle, inner sep = 3pt, fill=black, draw] {};

        \node (-2) at (-90pt,-60pt)     [black]   {$-2$};
        \node at (-60pt,-60pt) [circle, inner sep = 3pt, fill=black, draw] {};
        \node at (-30pt,-60pt) [circle, inner sep = 3pt, fill=black, draw] {};
        \node at (0pt,-60pt) [circle, inner sep = 3pt, fill=black, draw] {};
        \node at (30pt,-60pt) [circle, inner sep = 3pt, fill=black, draw] {};
        \node at (60pt,-60pt) [circle, inner sep = 3pt, fill=black, draw] {};

        \draw [<->] (-75pt,-120pt) -- (-75pt,135pt);
        \draw [-] (-75pt,-105pt) -- (75pt,-105pt);
        \node (h) at (-75pt,140pt) [black] {$h$};
        \node (t) at (80pt,-105pt) [black] {$t$};
    \end{tikzpicture} 
    \end{center}
    \end{subfigure}
    \begin{subfigure}{0.5\textwidth}
    \begin{center}
    \begin{tikzpicture}
        \tikzset{
            dots/.pic={
                \node at (0pt,5pt) [circle, inner sep = 0.5pt, fill=black, draw] {};
                \node at (0pt,0pt) [circle, inner sep = 0.5pt, fill=black, draw] {};
                \node at (0pt,-5pt) [circle, inner sep = 0.5pt, fill=black, draw] {};
            }
        }

        \tikzset{
            cross/.pic={
                \draw (-3pt,-3pt) edge [thick] (3pt,3pt);
                \draw (-3pt,3pt) edge [thick] (3pt,-3pt);
            }
        }

        \node (1) at (-60pt,-120pt)     [black]   {$1$};
        \node (2) at (-30pt,-120pt)     [black]   {$2$};
        \node (3) at (0pt,-120pt)     [black]   {$3$};
        \node (4) at (30pt,-120pt)     [black]   {$4$};
        \node (5) at (60pt,-120pt)     [black]   {$5$};

        \pic at (-90pt, 120pt) {dots};
        \pic at (-60pt, 120pt) {dots};
        \pic at (-30pt, 120pt) {dots};
        \pic at (0pt, 120pt) {dots};
        \pic at (30pt, 120pt) {dots};
        \pic at (60pt, 120pt) {dots};

        \pic at (-90pt, -90pt) {dots};
        \pic at (-60pt, -90pt) {dots};
        \pic at (-30pt, -90pt) {dots};
        \pic at (0pt, -90pt) {dots};
        \pic at (30pt, -90pt) {dots};
        \pic at (60pt, -90pt) {dots};

        \node (3) at (-90pt,90pt)     [black]   {$3$};
        \node at (-60pt,90pt) [circle, inner sep = 3pt, fill=black, draw] {};
        \node at (-30pt,90pt) [circle, inner sep = 3pt, fill=black, draw] {};
        \node at (0pt,90pt) [circle, inner sep = 3pt, fill=black, draw] {};
        \node at (30pt,90pt) [circle, inner sep = 3pt, fill=black, draw] {};
        \node at (60pt,90pt) [circle, inner sep = 3pt, fill=black, draw] {};
        
        \node (2) at (-90pt,60pt)     [black]   {$2$};
        \node at (-60pt,60pt) [circle, inner sep = 3pt, fill=black, draw] {};
        \pic at (-30pt,60pt) {cross};
        \node at (0pt,60pt) [circle, inner sep = 3pt, fill=black, draw] {};
        \node at (30pt,60pt) [circle, inner sep = 3pt, fill=black, draw] {};
        \pic at (60pt,60pt) {cross};

        \node (1) at (-90pt,30pt)     [black]   {$1$};
        \node at (-60pt,30pt) [circle, inner sep = 3pt, fill=black, draw] {};
        \pic at (-30pt,30pt) {cross} ;
        \pic at (0pt,30pt) {cross} ;
        \node at (30pt,30pt) [circle, inner sep = 3pt, fill=black, draw] {};
        \pic at (60pt,30pt) {cross} ;

        \node (0) at (-90pt,0pt)     [black]   {$0$};
        \node at (-60pt,0pt) [circle, inner sep = 3pt, fill=black, draw] {};
        \node at (-30pt,0pt) [circle, inner sep = 3pt, fill=black, draw] {};
        \node at (-0pt,0pt) [circle, inner sep = 3pt, fill=black, draw] {};
        \node at (30pt,0pt) [circle, inner sep = 3pt, fill=black, draw] {};
        \pic at (60pt,0pt) {cross} ;

        \node (-1) at (-90pt,-30pt)     [black]   {$-1$};
        \node at (-60pt,-30pt) [circle, inner sep = 3pt, fill=black, draw] {};
        \pic at (-30pt,-30pt) {cross} ;
        \node at (0pt,-30pt) [circle, inner sep = 3pt, fill=black, draw] {};
        \node at (30pt,-30pt) [circle, inner sep = 3pt, fill=black, draw] {};
        \node at (60pt,-30pt) [circle, inner sep = 3pt, fill=black, draw] {};

        \node (-2) at (-90pt,-60pt)     [black]   {$-2$};
        \node at (-60pt,-60pt) [circle, inner sep = 3pt, fill=black, draw] {};
        \node at (-30pt,-60pt) [circle, inner sep = 3pt, fill=black, draw] {};
        \node at (0pt,-60pt) [circle, inner sep = 3pt, fill=black, draw] {};
        \node at (30pt,-60pt) [circle, inner sep = 3pt, fill=black, draw] {};
        \node at (60pt,-60pt) [circle, inner sep = 3pt, fill=black, draw] {};

        \draw [<->] (-75pt,-120pt) -- (-75pt,135pt);
        \draw [-] (-75pt,-105pt) -- (75pt,-105pt);
        \node (h) at (-75pt,140pt) [black] {$h$};
        \node (t) at (80pt,-105pt) [black] {$t$};
    \end{tikzpicture}
    \end{center}
    \end{subfigure}

    \caption{Two illustrations of $\Lambda_5$ with the sets $B_{\vec t, \vec a, \vec b}^{2+}$ (left) and $B_{\vec t, \vec a, \vec b}^{2-}$ (right) marked with crosses, where $\vec t = (5,3,2,2)$, $\vec a=(0, 0, 1, -1)$, and $\vec b = (3, 2, 3, 0)$. Compare with Figure \ref{fig:lattice} and note the extra cross in the position $(3,2)$ on the left and the missing cross in the position $(3,0)$ on the right.}
    \label{fig:latticepm}
\end{figure}

\begin{theorem} \label{thm: ratio}
    Let K belong to the space $\mathcal{S}_d(\vec W)$, given in the form \eqref{def: kernelRHP}.  Let $B_{\vec t, \vec a, \vec b}$ be a finite subset of $\Lambda_N$, given by \eqref{def: gap} for some $q\in\mathbb N$. If $1-\lambda 1_{B_{\vec t, \vec a, \vec b}}K$ is invertible as an $\ell^2(\Lambda_N)$-operator, then the following identities hold:
    \begin{align}
    &\label{eq:ratioid1}
        \frac{\det\left(1 - \lambda 1_{B_{\vec t, \vec a, \vec b}^{j+}}K\right)}{\det\left(1 - \lambda 1_{B_{\vec t, \vec a, \vec b}}K\right)} = Y^{-T}(0;\vec t, \vec a, \vec b, \lambda)_{d+2j-1,d+2j-1}, \\ 
        &\label{eq:ratioid2}\frac{\det\left(1 - \lambda 1_{B_{\vec t, \vec a, \vec b}^{j-}}K\right)}{\det\left(1 - \lambda 1_{B_{\vec t, \vec a, \vec b}}K\right)} = Y^{-T}(0;\vec t, \vec a, \vec b, \lambda)_{d+2j,d+2j},
    \end{align}where $Y(\cdot;\vec t, \vec a, \vec b, \lambda)$ is the unique solution of RH problem \ref{prbl: main} with jump matrices $J^{(1)}$ and $J^{(2)}$ given by \eqref{def: J1}--\eqref{def: J2} and $Y^{-T}$ is the inverse transpose of $Y$.
\end{theorem}
\begin{remark}
By lowering the set $B_{t, a, b}$ repeatedly until it is empty, we can take a telescopic product of \eqref{eq:ratioid2} in order to obtain an expression for
$\det\left(1-\lambda K\right)_{\ell^2(B_{t,a,b})}$ as a product of entries of a sequence of RH problems, each corresponding to a subset of $B_{t,a,b}$.
\end{remark}
\begin{remark}\label{remark:inversetranspose}
Note that the ratio identities \eqref{eq:ratioid1}--\eqref{eq:ratioid2} provide local expressions in terms of the RH solution $Y$, evaluated at a single point $z=0$; in contrast to our expression \eqref{def:kernelresult} for the conditional correlation kernel, they do not involve any integration. It is also worth noting that there is no need to compute the inverse transpose of $Y$ to use these ratio identities. Indeed, $Y^{-T}$ satisfies the same RH problem \ref{prbl: main}, but with the jump matrices $J^{(1)}, J^{(2)}$ and the asymptotic matrix $Y^{(\infty)}$ replaced by their inverse transposes. Thus, one can use the ratio identities directly after analysing this modified RH problem.
\end{remark}
\begin{remark}
    If $B_{\vec t, \vec a, \vec b}$ is supported in a single time-slice of $\Lambda_N$, i.e. $t_1 = \dots = t_q$, Theorem \ref{thm: ratio} was already obtained in \cite[Theorem 4.2 with $\gamma_j = \lambda$]{charlier2025countingdominolozengetilings}.
    %It would be possible to combine their result with Theorem \ref{thm: main} in order to prove Theorem \ref{thm: ratio} as a corollary. However, since we choose to take different notation as what is used in \cite{charlier2025countingdominolozengetilings}, and we also wish to demonstrate the close connection between Theorems \ref{thm: main} and \ref{thm: ratio}, we have chosen to provide the proof of Theorem \ref{thm: ratio} in full.
\end{remark}

Our methodology to prove Theorems \ref{thm: main} and \ref{thm: ratio} is inspired by a method developed by Bertola and Cafasso in \cite{Bertola_2012} which allows us to characterise Fredholm determinants of double-contour integral kernels with RH problems. Their method is based on a clever decomposition of the associated operators involving the direct and inverse Fourier transforms. By conjugational invariance and other properties of the Fredholm determinant, this decomposition implies that the Fredholm determinant is equal to the Fredholm determinant of a simpler integrable kernel. The latter can then be related to a RH problem by the Its-Izergin-Korepin-Slavnov method \cite{IIKS}.
The method of \cite{Bertola_2012} was designed for continuous determinantal point processes and it has recently been adapted to a discrete setting in \cite{charlier2025countingdominolozengetilings}.
We apply the same techniques here, but generalise the results from \cite{charlier2025countingdominolozengetilings} in two ways: firstly, we extend to a multi-time setting; and secondly, we lift the whole method from the study of Fredholm determinants (or void probabilities) to the study of resolvent kernels (or conditional kernels of the point processes) which carry much more information about the underlying point processes.

\medskip

The single-time result from \cite{charlier2025countingdominolozengetilings}
was applied to the case of domino tilings of reduced Aztec diamonds in \cite{CharlierClaeys2025asymptotics}, and an asymptotic analysis of the RH problem led to strong asymptotics for the (weighted) number of domino tilings on such reduced domains. Our Theorem \ref{thm: main} gives access to more involved probabilistic quantities than only numbers of tilings, through the conditional correlation kernel. We intend to study asymptotics for finer microscopic correlations in this model in future work.

\medskip

More generally, our expression for the conditional correlation kernel provides a starting point for asymptotic analysis of general Schur processes on restricted domains. This situation is very similar to that of the Its-Izergin-Korepin-Slavnov method \cite{IIKS}, which provides a starting point to study asymptotics for conditional probabilities in determinantal point processes with integrable correlation kernels, and to that of the Bertola-Cafasso method \cite{Bertola_2012}, which provides a starting point to study void probabilities in determinantal point processes with specific double contour integral kernels.

\paragraph{Outline.}
We will prove Theorem \ref{thm: main} in Section \ref{section:proof1}. We will first prove that the conditional kernels admit a double-contour integral form, by using techniques similar to those of \cite{Bertola_2012, charlier2025countingdominolozengetilings}. Then we will use the theory of integrable operators developed in \cite{IIKS, DIZ, HarnadIts} to characterise the integrands by RH problems.
In Section \ref{section: ratio}, we will prove Theorem \ref{thm: ratio} as a consequence of Theorem \ref{thm: main}. 
Afterwards in Sections \ref{section:4} and \ref{section:5}, we will specialize our results to the case of domino tilings of reduced Aztec diamonds and lozenge tilings of hexagons with holes.

\section{The conditioned kernel}
\label{section:proof1}

\subsection{Reduction to $q=1$} \label{section 2.1}

We begin this section with the assumption that $K$ belongs to the space $\mathcal{S}_d(\vec W)$, given either in the form \eqref{def: kernel} or in the form \eqref{def: kernelRHP}. We consider $B_{\vec t, \vec a, \vec b}$ given by \eqref{def: gap}, for some $q\in\mathbb N$, and $\lambda \in \C$, such that $1-\lambda 1_{B_{\vec t, \vec a, \vec b}}K$ is invertible as an $\ell^2(\Lambda_N)$-operator. We make the observation that, for the proof of Theorem \ref{thm: main}, we can reduce to the case $q=1$. This follows from the fact that 
\begin{align}
    \left(K_{t_1,a_1,b_1, \lambda}\right)_{t_2,a_2,b_2,\lambda} &= K_{t_1, a_1, b_1,\lambda} \left(1 - \lambda 1_{B_{t_2, a_2, b_2}}K_{t_1, a_1, b_1}\right)^{-1} \nonumber \\ 
    & = K \left(1 - \lambda 1_{B_{t_1, a_1, b_1}}K\right)^{-1} \left[1 - \lambda 1_{B_{t_2, a_2, b_2}}K \left(1 - \lambda 1_{B_{t_1, a_1, b_1}}K\right)^{-1} \right]^{-1} \nonumber \\
    & = K \left(1 - \lambda 1_{B_{t_1, a_1, b_1}}K - \lambda 1_{B_{t_2, a_2, b_2}}K \right)^{-1} \nonumber \\
    &= K \left(1 - \lambda 1_{B_{t_1,t_2;a_1,a_2;b_1,b_2}}K\right)^{-1} \nonumber \\
    &= K_{t_1,t_2;a_1,a_2;b_1,b_2;\lambda}.
\end{align}
We therefore see that we can construct $K_{\vec t,\vec a,\vec b,\lambda}$ iteratively by adding one cluster $B_{t_j,a_j,b_j}$ at a time in the following manner:
\[K \mapsto K_{t_1,a_1,b_1,\lambda} \mapsto K_{t_1,t_2;a_1,a_2;b_1,b_2;\lambda} \mapsto \dots \mapsto K_{\vec t,\vec a,\vec b,\lambda}.\]

However, this procedure relies on the assumption that not only $1-\lambda 1_{B_{\vec t, \vec a, \vec b}}K$ is invertible, but also that, for any $p = 1, \dots, q$,
\[1 - \lambda \sum_{j=1}^{p} 1_{B_{t_j,a_j,b_j}}K\]
is invertible. This is true if $K$ defines a determinantal point process on $\Lambda_N$; the Fredholm determinants of each operator can be shown to be non-zero via the corresponding gap probabilities, but it is not true for general integral operators. To avoid the need for these additional assumptions, we define $A \subset \C$ as the set of $\mu \in \C$ such that, for each $p = 1, \dots, q$, 
\[1 - \mu \sum_{j=1}^{p} 1_{B_{t_j,a_j,b_j}}K\]
is invertible. We observe that $A$ is a dense (even finite) subset of $\C$. Then, supposing Theorem \ref{thm: main} is true for all $\mu \in A$, we observe that both sides of \eqref{def:kernelresult} are continuous in $\lambda$, for $\lambda$ such that $1 - \lambda 1_{B_{\vec t, \vec a, \vec b}}K$ is invertible, and we conclude that \eqref{def:kernelresult} extends to all these values of $\lambda$, even if they do not belong to $A$.
Thus, we can proceed with the proof of Theorem \ref{thm: main} assuming that $\lambda \in A$, and moreover, it suffices to let $q = 1$.

\medskip

Indeed, to complete the proof of Theorem \ref{thm: main}, given it holds for $q=1$, it is simply necessary to verify the formulas \eqref{def: J1}--\eqref{def: phipsi}. For $q=1$, these translate to
\begin{equation} \label{def: J1b}
    J^{(1)}(z;t_1,a_1,b_1, \lambda) = 
    \left[\begin{array}{c|c}
        J^{(1)}(z) & \vec \psi_{t_1}(z)W_{t_1}(z) \vec f_{1}(z)^T  \\ \hline 0_{2 \times d} & I_2
    \end{array}\right],
\end{equation}
\begin{equation} \label{def: J2n-b}
    J^{(2)}(z;t_1,a_1,b_1, \lambda) =
    \left[\begin{array}{c|c}
        J^{(2)}(z) & 0_{d \times 2} \\ \hline
        \vec g_1(z) \frac{1}{W_{t_1}(z)} \vec \phi_{t_1}(z)^T & I_2
    \end{array}\right],
\end{equation}
\begin{equation} \label{def: prephipsi}
    \vec \phi_t(u;t_1,a_1,b_1, \lambda) := \begin{bmatrix}
        \vec \phi_t(u) \\ -1_{t>t_1} \vec f_1(u) W_{t_1}(u)
    \end{bmatrix}, \quad \vec \psi_{t'}(v;t_1,a_1,b_1, \lambda) := \begin{bmatrix}
        \vec \psi_{t'}(v) \\ -1_{t_1>t'} \vec g_1(v) \frac{1}{W_{t_1}(v)}
    \end{bmatrix},
\end{equation}
and it follows that, by iterating \eqref{def: J1b}--\eqref{def: prephipsi} for each additional cluster $B_{t_j,a_j,b_j}$, we do indeed retain the formulas \eqref{def: J1}--\eqref{def: phipsi}. The remainder of Section \ref{section:proof1} is concerned with the computation of $K_{t_1,a_1,b_1,\lambda}$.

\medskip

\subsection{Matrix representation of operators}\label{subsec:matrix}
For the computation of $K_{t_1,a_1,b_1,\lambda}(t,h;t',h')$, we first assume that $t$, $t'$, and $t_1$ are all distinct; an assumption which we will be able to remove later. Consider the smaller lattice $\Lambda_{t,t',t_1} := \{t,t',t_1\} \times \mathbb{Z}$. From the elementary identity
\begin{equation}
    (1-BA)^{-1}B=B(1-AB)^{-1},
\end{equation}
applied to $A=\lambda 1_{B_{t_1,a_1,b_1}}K$ and $B=1_{\Lambda_{t,t',t_1}}$, we obtain 
\begin{equation}
1_{\Lambda_{t,t',t_1}}K\big(1-\lambda 1_{B_{t_1, a_1, b_1}}K\big)^{-1}1_{\Lambda_{t,t',t_1}} = 1_{\Lambda_{t,t',t_1}}K1_{\Lambda_{t,t',t_1}}\big(1-\lambda 1_{B_{t_1, a_1, b_1}}1_{\Lambda_{t,t',t_1}}K1_{\Lambda_{t,t',t_1}}\big)^{-1}.
\end{equation}
Hence it is sufficient to compute the kernel of $K_{t_1,a_1,b_1,\lambda}$ as an operator acting on $\ell^2\left(\Lambda_{t,t',t_1}\right)$.

\medskip

There is a natural isomorphism between the spaces $\ell^2(\Lambda_{t,t',t_1})$ and $\ell^2\left(\Z;\C^3\right)$, the space of $3\times 1$ vectors whose entries are sequences in $\ell^2(\mathbb Z)$. This gives rise to a $3\times 3$ matrix representation of operators which will be convenient in what follows. We denote this isomorphism by $\Phi: \ell^2\big(\Z;\C^3\big) \rightarrow \ell^2\big(\Lambda_{t,t',t_1}\big)$, defined for $\sigma \in \ell^2\big(\Z;\C^3\big)$ by
\begin{equation}
  \label{def:Phi}   
  \Phi[\sigma](t,h) = \sigma(h)_1, \quad
    \Phi[\sigma](t',h) = \sigma(h)_2, \quad
    \Phi[\sigma](t_1,h) = \sigma(h)_3.
\end{equation}
Under this transformation, $\widehat{K} := \Phi^{-1} K \Phi$ has matrix-valued kernel given by
\begin{equation}\label{def:hatK}
    \widehat{K}(h,h') = \begin{pmatrix}
        K(t,h;t,h') & K(t,h;t',h') & K(t,h;t_1,h') \\ 
        K(t',h;t,h') & K(t',h;t',h') & K(t',h;t_1,h') \\ 
        K(t_1,h;t,h') & K(t_1,h;t',h') & K(t_1,h;t_1,h')
    \end{pmatrix}.
\end{equation}
By \eqref{def: kernel}, we can write
    \begin{equation}\label{def:hatKdoubleintegral}
        \widehat{K}(h,h') = -\oint_{\Sigma_1} z^{-h}V(z)z^{h'} \frac{dz}{2\pi iz} + \oint_{\Sigma_2} \oint_{\Sigma_1} u^{-h} Q(u,v) v^{h'} \frac{dv}{2\pi iv} \frac{du}{2\pi iu},
    \end{equation}
where
\begin{equation}
    V(z) := \begin{pmatrix} \label{def: V}
        0 & 1_{t>t'}\frac{W_{t'}(z)}{W_{t}(z)} & 1_{t>t_1}\frac{W_{t_1}(z)}{W_{t}(z)} \\
        1_{t'>t}\frac{W_{t}(z)}{W_{t'}(z)} & 0 & 1_{t'>t_1}\frac{W_{t_1}(z)}{W_{t'}(z)} \\
        1_{t_1>t}\frac{W_{t}(z)}{W_{t_1}(z)} & 1_{t_1>t'}\frac{W_{t'}(z)}{W_{t_1}(z)} & 0 
    \end{pmatrix},
\end{equation}
and where $Q$ is the following $(d\times 3)$-integrable kernel,
\begin{equation} \label{def: Q}
    Q(u,v) := \frac{v}{u-v}\begin{pmatrix}
        \frac{1}{W_{t}(u)} \vec \rho_{t}(u)^T \\
        \frac{1}{W_{t'}(u)} \vec \rho_{t'}(u)^T \\
        \frac{1}{W_{t_1}(u)} \vec \rho_{t_1}(u)^T
    \end{pmatrix}  \begin{pmatrix}
        \vec \sigma_{t}(v)W_{t}(v) & \vec \sigma_{t'}(v)W_{t'}(v) & \vec \sigma_{t_1}(v)W_{t_1}(v)
    \end{pmatrix}.
\end{equation}
Likewise, $\widehat{K}_{t_1,a_1,b_1,\lambda} := \Phi^{-1} K_{t_1,a_1,b_1,\lambda} \Phi$ has matrix-valued kernel given by
\begin{equation} \label{expr: matrix KB hat}
    \widehat{K}_{t_1,a_1,b_1,\lambda}(h,h') = \begin{pmatrix}
        K_{t_1,a_1,b_1,\lambda}(t,h;t,h') & K_{t_1,a_1,b_1,\lambda}(t,h;t',h') & K_{t_1,a_1,b_1,\lambda}(t,h;t_1,h') \\ 
        K_{t_1,a_1,b_1,\lambda}(t',h;t,h') & K_{t_1,a_1,b_1,\lambda}(t',h;t',h') & K_{t_1,a_1,b_1,\lambda}(t',h;t_1,h') \\ 
        K_{t_1,a_1,b_1,\lambda}(t_1,h;t,h') & K_{t_1,a_1,b_1,\lambda}(t_1,h;t',h') & K_{t_1,a_1,b_1,\lambda}(t_1,h;t_1,h')
    \end{pmatrix}.
\end{equation}
The definition \eqref{def: Ktransform} then translates to
\begin{equation}
    \widehat{K}_{t_1,a_1,b_1,\lambda} = \widehat{K} \left(1-\lambda\widehat{1}_{B_{t_1, a_1, b_1}} \widehat{K} \right)^{-1},
\end{equation}
where $\widehat{1}_{B_{t_1, a_1, b_1}} := \Phi^{-1} 1_{B_{t_1, a_1, b_1}} \Phi$ acts by left multiplication:
\begin{equation}  \label{def:1B}\widehat{1}_{B_{t_1, a_1, b_1}}[\sigma](h) = \begin{pmatrix}
        0 & 0 & 0 \\
        0 & 0 & 0 \\
        0 & 0 & 1_{[a_1,b_1)}(h)
    \end{pmatrix}\sigma(h), \quad \text{for $\sigma \in \ell^2\left(\Z;\C^3\right)$.}
\end{equation}

\subsection{Fourier conjugation}\label{subsec:Fourier}
    In our analysis, we make use of Fourier series transforms. We write $\mathcal{F}: L^2(\Sigma_1) \rightarrow \ell^2(\mathbb{Z})$ for the Fourier series transform defined by
\begin{equation}\label{def:Fourier}
    \mathcal{F}[f](h) := \int_{\Sigma_1} z^{-h} f(z) \frac{dz}{2\pi i z}, \quad h \in \mathbb{Z},
\end{equation}
and $\mathcal{F}^{-1} : \ell^2(\mathbb{Z}) \rightarrow L^2(\Sigma_1)$ for the inverse given by
\begin{equation}\label{def:Fourierinverse}
    \mathcal{F}^{-1}[\sigma](z) := \sum_{h\in\mathbb{Z}} z^h\sigma(h), \quad z\in \Sigma_1.
\end{equation}
We will let $\mathcal F$ and $\mathcal{F}^{-1}$ act component-wise on vector-valued and matrix-valued functions. In this way we can express
\begin{equation} \label{expr: fourierK}
    \widehat{K}(h,h') = - \F \left[z^{-h}V(z)\right](-h') + \F \left[ \oint_{\Sigma_2} u^{-h} Q(u,z) \frac{du}{2\pi iu}\right](-h').
\end{equation}

\medskip

We define $L$, acting on $L^2(\Sigma_1;\C^3)$, by
\begin{equation}\label{def:L}
    L := \F^{-1} \lambda \widehat{1}_{B_{t_1, a_1, b_1}} \widehat{K}\F.
\end{equation}

\begin{lemma}
\label{lemma:L}The operator $L:L^2(\Sigma_1,\C^3)\to L^2(\Sigma_1,\C^3)$ can be written as
\begin{equation}
    L = - P^{11}V + P^{12}Q,
\end{equation}
where, for
\begin{equation} \label{def: P}
        P(z,w) := \lambda \frac{w}{z-w}\begin{pmatrix}
        0 & 0 & 0 \\
        0 & 0 & 0 \\
        0 & 0 & (z/w)^{b_1} - (z/w)^{a_1}
    \end{pmatrix},
    \end{equation}
and for $f_1 \in L^2(\Sigma_1;\C^3)$ and $f_2 \in L^2(\Sigma_2;\C^3)$,
\begin{align}
    &P^{11}V : L^2(\Sigma_1;\C^3) \rightarrow L^2(\Sigma_1;\C^3),  &&P^{11}V[f_1](z) = \oint_{\Sigma_1} P(z,w)V(w)f_1(w) \frac{dw}{2\pi iw}, \\
    &P^{12} : L^2(\Sigma_2;\C^3) \rightarrow L^2(\Sigma_1;\C^3),  &&P^{12}[f_2](z) = \oint_{\Sigma_2} P(z,w)f_2(w) \frac{dw}{2\pi iw}, \\
    &Q : L^2(\Sigma_1;\C^3) \rightarrow L^2(\Sigma_2;\C^3), &&Q[f_1](z) = \oint_{\Sigma_2} Q(z,w)f_1(w) \frac{dw}{2\pi iw}.
\end{align}
\end{lemma}
\begin{proof}
We compute the kernel of $L$ by letting it act on an arbitrary function $f \in L^2(\Sigma_1,\C^3)$. We do this in steps: Firstly,
\begin{align}
    &\widehat{K}\F[f](h) = \sum_{h' \in \Z} \widehat{K}(h,h') \F[f](h') \nonumber \\
    &= \sum_{h' \in \Z} \left( - \F \left[z^{-h}V(z)\right](-h') + \F \left[ \oint_{\Sigma_2} u^{-h} Q(u,z) \frac{du}{2\pi iu}\right](-h') \right) \F[f](h') \nonumber \\
    &= -\oint_{\Sigma_1} w^{-h} V(w) f(w) \frac{dw}{2\pi iw} + \oint_{\Sigma_2} \oint_{\Sigma_1} u^{-h} Q(u,w) f(w) \frac{dw}{2\pi iw} \frac{du}{2\pi iu}.
\end{align}
Now,
\begin{equation}
    L[f](z) = \lambda \sum_{h \in \mathbb{Z}} z^h \begin{pmatrix}
        0 & 0 & 0 \\
        0 & 0 & 0 \\
        0 & 0 & 1_{[a_1,b_1)}(h) 
    \end{pmatrix} \widehat{K}\F[f](h).
\end{equation}
We note that there is a finite number of non-zero terms in this sum and that
\begin{equation}
    \lambda \sum_{h \in \mathbb{Z}} z^h \begin{pmatrix}
        0 & 0 & 0 \\
        0 & 0 & 0 \\
        0 & 0 & 1_{[a_1,b_1)}(h) 
    \end{pmatrix} w^{-h} = \lambda \frac{w}{z-w}\begin{pmatrix}
        0 & 0 & 0 \\
        0 & 0 & 0 \\
        0 & 0 & (z/w)^{b_1} - (z/w)^{a_1}
    \end{pmatrix}= P(z,w).
\end{equation}
Therefore,
\begin{equation}
    L[f](z) = -\oint_{\Sigma_1} P(z,w)V(w)f(w) \frac{dw}{2\pi iw} + \oint_{\Sigma_2} \oint_{\Sigma_1} P(z,u) Q(u,w) f(w) \frac{dw}{2\pi iw} \frac{du}{2\pi iu},
\end{equation}
and the result follows.
\end{proof}

\subsection{Towards an integrable kernel}\label{subsec:toint}
Using \eqref{def:L}, we obtain
\begin{equation}
\left(1-\lambda \widehat 1_{B_{t_1,a_1,b_1}}\widehat K\right)^{-1}=\mathcal F\left(1-L\right)^{-1}\mathcal F^{-1}.\label{eq:Lresolvent}
\end{equation}
We calculate $(1-L)^{-1}$ following a method developed in \cite{Bertola_2012}. By considering $(1-L)^{-1}$ in the larger space $L^2(\Sigma_1;\C^3) \oplus L^2(\Sigma_2;\C^3)$ we have
\begin{align}
        \left(1 - L\right)^{-1} &= \begin{bmatrix}
        1 & 0
    \end{bmatrix} \begin{bmatrix}
        1 + P^{11}V- P^{12}Q & 0 \\ 0 & 1
    \end{bmatrix}^{-1} \begin{bmatrix}
        1 \\ 0
    \end{bmatrix} \nonumber\\
    &= \begin{bmatrix}
        1 & 0
    \end{bmatrix} \begin{bmatrix}
        1 & 0 \\ -Q & 1
    \end{bmatrix} \begin{bmatrix}
        1 + P^{11}V & -P^{12} \\ -Q & 1
    \end{bmatrix}^{-1} \begin{bmatrix}
        1 & -P^{12} \\ 0 & 1
    \end{bmatrix} \begin{bmatrix}
        1 \\ 0
    \end{bmatrix} \nonumber\\
    &= \begin{bmatrix}
        1 & 0
    \end{bmatrix} \begin{bmatrix}
        1 + P^{11}V & -P^{12} \\ -Q & 1
    \end{bmatrix}^{-1} \begin{bmatrix}
        1 \\ 0
    \end{bmatrix},\label{eq:resolvent1}
    \end{align}
    where, for $f_1 \in L^2(\Sigma_1;\C^3)$ and $f_2 \in L^2(\Sigma_2;\C^3)$,
    \begin{align}
        \begin{bmatrix}
        1 & 0
    \end{bmatrix} &: L^2(\Sigma_1;\C^3) \oplus L^2(\Sigma_2;\C^3) \rightarrow L^2(\Sigma_1;\C^3), \quad & \begin{bmatrix}
        1 & 0
    \end{bmatrix} \left(f_1 \oplus f_2\right) = f_1, \\
    \begin{bmatrix}
        1 \\ 0
    \end{bmatrix} &: L^2(\Sigma_1;\C^3) \rightarrow L^2(\Sigma_1;\C^3) \oplus L^2(\Sigma_2;\C^3), \quad & \begin{bmatrix}
        1 \\ 0
    \end{bmatrix} \left(f_1\right) = f_1 \oplus 0.
    \end{align}
    Since the contours $\Sigma_1$ and $\Sigma_2$ are disjoint, there is a natural isomorphism between $L^2 \left(\Sigma_1;\C^3\right) \oplus L^2 \left(\Sigma_2;\C^3\right)$ and $L^2\left(\Sigma_1 \cup \Sigma_2;\C^3\right)$. We denote this isomorphism by $$\Psi: L^2\left(\Sigma_1;\C^3\right) \oplus L^2 \left(\Sigma_2;\C^3\right) \rightarrow L^2\left(\Sigma_1 \cup \Sigma_2;\C^3\right),$$ given, for $f_1 \oplus f_2 \in L^2\left(\Sigma_1;\C^3\right) \oplus L^2 \left(\Sigma_2;\C^3\right)$, by
    \begin{equation} \label{def:Psi}
        \Psi [f_1 \oplus f_2](z) = \begin{cases}
            f_1(z), \quad \text{for $z \in \Sigma_1$,} \\
            f_2(z), \quad \text{for $z \in \Sigma_2$.}
        \end{cases}
    \end{equation}
    We define
    \begin{equation}\label{def:M}
        M := \Psi \begin{bmatrix}
        - P^{11}V & P^{12} \\ Q & 0
    \end{bmatrix} \Psi^{-1}.
    \end{equation}
This operator acts on $L^2\left(\Sigma_1 \cup \Sigma_2;\C^3\right)$ with matrix-valued kernel
    \begin{multline} \label{def: M}
        M(z,w) = -1_{\Sigma_1}(z)1_{\Sigma_1}(w)P(z,w)V(w) \\ + 1_{\Sigma_1}(z)1_{\Sigma_2}(w)P(z,w) + 1_{\Sigma_2}(z)1_{\Sigma_1}(w)Q(z,w).
    \end{multline}
Crucially, the kernel $M$ is of $(d+2)\times 3$-integrable type, given by \eqref{def: integrable kernel}, in the setting $\Sigma = \Sigma_1 \cup \Sigma_2$. The associated matrix-valued functions $F$ and $G$ are given in the form \eqref{def: F12}--\eqref{def: G12} with
\begin{align}
    F_1(z) &= \left[\begin{array}{c c c}
                \vec 0_d & \vec 0_d & \vec 0_d \\ \hline  
                \vec 0_2 & \vec 0_2 & \vec f_{1}(z)
            \end{array}\right],  \label{def: F1} \\
    F_2(z) &= \left[\begin{array}{c c c}
                \vec \rho_{t}(z) \frac{1}{W_{t}(z)} & \vec \rho_{t'}(z) \frac{1}{W_{t'}(z)} & \vec \rho_{t_1}(z) \frac{1}{W_{t_1}(z)} \\ \hline
                \vec 0_2 & \vec 0_2 & \vec 0_2
            \end{array}\right], \label{def: F2} \\
    G_1(w) &= \left[\begin{array}{c c c}
               \vec\sigma_{t}(w)W_{t}(w) & \vec\sigma_{t'}(w)W_{t'}(w) & \vec\sigma_{t_1}(w)W_{t_1}(w) \\ \hline
              -1_{t_1>t} \vec g_{1}(w) \frac{W_{t}(w)}{W_{t_1}(w)} & -1_{t_1>t'} \vec g_{1}(w) \frac{W_{t'}(w)}{W_{t_1}(w)} & \vec 0_2  
            \end{array}\right], \label{def: G1} \\
    G_2(w) &= \left[\begin{array}{ccc}
                \vec 0_d & \vec 0_d & \vec 0_d \\ \hline
                \vec 0_2 & \vec 0_2 & \vec g_{1}(w)
            \end{array}\right]. \label{def: G2}
\end{align}

\subsection{Its-Izergin-Korepin-Slavnov theory}\label{subsec:IIKS}

We introduce the resolvent operator
\begin{equation} \label{def: R}
    R := M\left(1-M\right)^{-1},
\end{equation}
for $M$, given by \eqref{def:M}, and apply the theory of IIKS summarised in Section \ref{intro}. Since $R$ is of the form \eqref{def: resolvent}, we see that $R$ is also of $(d+2) \times 3$-integrable type, given by \eqref{def: Rintegrable}--\eqref{def: Gtilde}. In this case, RH problem \ref{RHPUpsilonIIKS} becomes the following:
\begin{problem} \label{prbl: IIKS}
    \begin{enumerate}
        \item $\Upsilon: \mathbb{C} \setminus (\Sigma_1 \cup \Sigma_2) \rightarrow \mathbb{C}^(d+2)\times (d+2)$ is analytic.
        \item 
        For $z \in \Sigma_1 \cup \Sigma_2$, $\Upsilon$ has continuous boundary values, $\Upsilon_+$ and $\Upsilon_-$, and they satisfy the jump relations
        \begin{align}
            \Upsilon_+(z) &= \Upsilon_-(z)\left[\begin{array}{c|c}
                I_d & \vec \sigma_{t_1}(z)W_{t_1}(z) \vec f_1(z)^T \\ \hline
                \vec 0_{2 \times d} & I_2
            \end{array}\right], \quad \text{for $z \in \Sigma_1$,} \\
            \Upsilon_+(z) &= \Upsilon_-(z)\left[\begin{array}{c|c}
                I_d & \vec 0_{d \times 2} \\ \hline
                \vec g_1(z)^T \frac{1}{W_{t_1}(z)}\vec \rho_{t_1}(z) & I_2
            \end{array}\right], \quad \text{for $z \in \Sigma_2$.}
        \end{align}
        \item $\Upsilon(z) = I_{d+2} + \mathcal{O}(1/z)$ as $z \rightarrow \infty$.
    \end{enumerate}
    \end{problem}

\medskip

Combining \eqref{eq:resolvent1}, \eqref{def:M} and \eqref{def: R}, we have
    \begin{align}
        \left(1 - L\right)^{-1} &= \begin{bmatrix}
        1 & 0
    \end{bmatrix} \Psi^{-1} (1-M)^{-1} \Psi \begin{bmatrix}
        1 \\ 0
    \end{bmatrix} \nonumber\\
    &= \begin{bmatrix}
        1 & 0
    \end{bmatrix} \Psi^{-1} (1+R) \Psi \begin{bmatrix}
        1 \\ 0
    \end{bmatrix} \nonumber\\
    &= 1 + \begin{bmatrix}
        1 & 0
    \end{bmatrix} \Psi^{-1} R \Psi \begin{bmatrix}
        1 \\ 0
    \end{bmatrix},
    \end{align}
    allowing us to describe the action of $\left(1 - L\right)^{-1}$. For $f \in L^2(\Sigma_1)$ and $z \in \Sigma_1$,
\begin{equation}\label{eq:Linv}
        \left(1 - L\right)^{-1}[f](z) = f(z) + \oint_{\Sigma_1} R(z,w)f(w) \frac{dw}{2\pi iw}.
    \end{equation}
    This prepares us to calculate the kernel $\widehat K_{t_1,a_1,b_1,\lambda}$. 
    \begin{lemma}
    We have
        \begin{multline} \label{eq:Khatidentity} 
        \widehat{K}_{t_1,a_1,b_1, \lambda}(h,h') = - \oint_{\Sigma_1} z^{-h}V(z)z^{h'}\frac{dz}{2\pi iz} \\
        + \oint_{\Sigma_2} \oint_{\Sigma_1} z^{-h} \frac{w}{z-w}\left(F_2(z)^T-V(z)F_1(z)^T\right)\Upsilon_\pm(z)^{-1}\Upsilon_\pm(w)G_1(w)^Tw^{h'} \frac{dw}{2\pi iw} \frac{dz}{2\pi iz}.
    \end{multline}
    \end{lemma}
    \begin{proof}
    By \eqref{eq:Lresolvent} and \eqref{eq:Linv}, we have, for $\sigma \in L^2(\Z;\C^3)$, 
    \begin{align}
        \left(1-\lambda \widehat{1}_{B_{t_1, a_1, b_1}} \widehat{K} \right)^{-1}[\sigma](h) &= \sigma(h) + \F \left[\oint_{\Sigma_1} R(z,w) \F^{-1}[\sigma](w) \frac{dw}{2\pi iw} \right](h),
    \end{align}
    and therefore
    \begin{align}
        &\widehat K_{t_1,a_1,b_1,\lambda}[\sigma](h) = \widehat{K}[\sigma](h) + \sum_{h' \in \Z} \widehat{K}(h,h') \F \left[\oint_{\Sigma_1} R(z,w) \F^{-1}[\sigma](w) \frac{dw}{2\pi iw} \right](h'). \label{2.47}
    \end{align}
    We compute the second term in two parts:
    \begin{equation}\label{eq:A12}
        \sum_{h' \in \Z} \widehat{K}(h,h') \F \left[\oint_{\Sigma_1} R(z,w) \F^{-1}[\sigma](w) \frac{dw}{2\pi iw} \right](h') = A_1(h)+A_2(h),
       \end{equation}
        where
        \begin{align}
        &A_1(h)=
        - \sum_{h' \in \Z} \F \left[z^{-h}V(z)\right](-h') \F \left[\oint_{\Sigma_1} R(z,w) \F^{-1}[\sigma](w) \frac{dw}{2\pi iw} \right](h'),\\
        &A_2(h)= \sum_{h' \in \Z} \F \left[ \oint_{\Sigma_2} u^{-h} Q(u,z) \frac{du}{2\pi iu}\right](-h') \F \left[\oint_{\Sigma_1} R(z,w) \F^{-1}[\sigma](w) \frac{dw}{2\pi iw} \right] (h').
    \end{align}
For the first part,
    \begin{align}
        A_1(h)&= -\oint_{\Sigma_1} \oint_{\Sigma_1} u^{-h}V(u)R(u,v) \F^{-1}[\sigma](v) \frac{dv}{2\pi iv} \frac{du}{2\pi iu} \nonumber\\
        &= -\sum_{h' \in \Z} \oint_{\Sigma_1} \oint_{\Sigma_1} u^{-h}V(u)R(u,v) v^{h'} \frac{dv}{2\pi iv} \frac{du}{2\pi iu} \sigma(h'), \label{eq:A1}
    \end{align}
    and for the second part
    \begin{align}
        A_2(h)&= \oint_{\Sigma_2} \oint_{\Sigma_1} \oint_{\Sigma_1} u^{-h}Q(u,v)R(v,w) \F^{-1}[\sigma](w) \frac{dw}{2\pi iw}\frac{dv}{2\pi iv}\frac{du}{2\pi iu} \nonumber\\
        &= \sum_{h' \in \Z} \oint_{\Sigma_2} \oint_{\Sigma_1} \oint_{\Sigma_1} u^{-h}Q(u,v)R(v,w) w^{h'} \frac{dw}{2\pi iw}\frac{dv}{2\pi iv}\frac{du}{2\pi iu} \sigma(h').\label{eq:A2}
    \end{align}
    To simplify \eqref{eq:A2}, we note that, for $u \in \Sigma_2$ and $w \in \Sigma_1$, $$\oint_{\Sigma_1} Q(u,z)R(z,w) \frac{dz}{2\pi iz}$$ is the integral kernel of the operator
    \[Q \begin{bmatrix}
        1 & 0
    \end{bmatrix} \Psi^{-1} R \Psi \begin{bmatrix}
        1 \\ 0
    \end{bmatrix} : L^2\left(\Sigma_1;\C^3 \right) \rightarrow L^2\left(\Sigma_2;\C^3 \right), \]
    and we have that
    \begin{align}
        Q \begin{bmatrix}
        1 & 0
    \end{bmatrix} \Psi^{-1} R \Psi \begin{bmatrix}
        1 \\ 0
    \end{bmatrix} &= -Q + Q \begin{bmatrix}
        1 & 0
    \end{bmatrix} \Psi^{-1} (1+R) \Psi \begin{bmatrix}
        1 \\ 0
    \end{bmatrix} \nonumber \\
    &= -Q + Q \begin{bmatrix}
        1 & 0
    \end{bmatrix} \Psi^{-1} (1-M)^{-1} \Psi \begin{bmatrix}
        1 \\ 0
    \end{bmatrix}\nonumber\\
    &= -Q + Q \begin{bmatrix}
        1 & 0
    \end{bmatrix} \begin{bmatrix}
        1 + P^{11}V & -P^{12} \\ -Q & 1
    \end{bmatrix}^{-1} \begin{bmatrix}
        1 \\ 0
    \end{bmatrix} \nonumber\\
    &= -Q + \begin{bmatrix}
        0 & 1
    \end{bmatrix} \begin{bmatrix}
        -P^{11}V & P^{12} \\ Q & 0
    \end{bmatrix} \begin{bmatrix}
        1 + P^{11}V & -P^{12} \\ -Q & 1
    \end{bmatrix}^{-1} \begin{bmatrix}
        1 \\ 0
    \end{bmatrix} \nonumber\\
    &= -Q + \begin{bmatrix}
        0 & 1
    \end{bmatrix} \Psi^{-1} M(1-M)^{-1} \Psi \begin{bmatrix}
        1 \\ 0
    \end{bmatrix} \nonumber\\
    &= -Q + \begin{bmatrix}
        0 & 1
    \end{bmatrix} \Psi^{-1} R \Psi \begin{bmatrix}
        1 \\ 0
    \end{bmatrix}.
    \end{align}
    Therefore, for $u \in \Sigma_2$ and $w \in \Sigma_1$,
    \begin{equation} \label{expr: QR=-Q+R}
        \oint_{\Sigma_1} Q(u,v)R(v,w) \frac{dv}{2\pi iv} = -Q(u,w) + R(u,w).
    \end{equation}
    Substituting \eqref{expr: QR=-Q+R} into \eqref{eq:A2}, we obtain
    \begin{multline} \label{eq:A22}
        A_2(h)=
        -\sum_{h' \in \Z} \oint_{\Sigma_2} \oint_{\Sigma_1} u^{-h}Q(u,v) v^{h'} \frac{dv}{2\pi iv} \frac{du}{2\pi iu} \sigma(h') \\
        +\sum_{h' \in \Z} \oint_{\Sigma_2} \oint_{\Sigma_1} u^{-h}R(u,v) v^{h'} \frac{dv}{2\pi iv} \frac{du}{2\pi iu} \sigma(h'),
    \end{multline}
    and substituting \eqref{eq:A1} and \eqref{eq:A22} into \eqref{eq:A12}, and then into \eqref{2.47}, we obtain
    \begin{multline}
        \widehat{K}_{t_1,a_1,b_1,\lambda}[\sigma](h) = \widehat{K}[\sigma](h) - \sum_{h' \in \Z} \oint_{\Sigma_1} \oint_{\Sigma_1} u^{-h}V(u)R(u,v) v^{h'} \frac{dv}{2\pi iv} \frac{du}{2\pi iu} \sigma(h') \\
        -\sum_{h' \in \Z} \oint_{\Sigma_2} \oint_{\Sigma_1} u^{-h}Q(u,v) v^{h'} \frac{dv}{2\pi iv} \frac{du}{2\pi iu} \sigma(h') + \sum_{h' \in \Z} \oint_{\Sigma_2} \oint_{\Sigma_1} u^{-h}R(u,v) v^{h'} \frac{dv}{2\pi iv} \frac{du}{2\pi iu} \sigma(h'),
    \end{multline}
    which simplifies to
    \begin{multline}
        \widehat{K}_{t_1,a_1,b_1,\lambda}[\sigma](h) = - \sum_{h' \in \Z}  \oint_{\Sigma_1} z^{-h}V(z)z^{h'} \frac{dz}{2\pi iz} \sigma(h')
        \\ -\sum_{h' \in \Z} \oint_{\Sigma_1} \oint_{\Sigma_1} u^{-h}V(u)R(u,v) v^{h'} \frac{dv}{2\pi iv} \frac{du}{2\pi iu} \sigma(h') \\ + \sum_{h' \in \Z} \oint_{\Sigma_2} \oint_{\Sigma_1} u^{-h}R(u,v) v^{h'} \frac{dv}{2\pi iv} \frac{du}{2\pi iu} \sigma(h').
    \end{multline}
    Therefore,
    \begin{multline}
        \widehat{K}_{t_1,a_1,b_1,\lambda}(h,h') = - \oint_{\Sigma_1} z^{-h}V(z)z^{h'}\frac{dz}{2\pi iz} 
        -\oint_{\Sigma_1} \oint_{\Sigma_1} u^{-h}V(u)R(u,v) v^{h'} \frac{dv}{2\pi iv} \frac{du}{2\pi iu} \\
        + \oint_{\Sigma_2} \oint_{\Sigma_1} u^{-h}R(u,v) v^{h'} \frac{dv}{2\pi iv} \frac{du}{2\pi iu}.
    \end{multline}
    Now by \eqref{def: Rintegrable}--\eqref{def: Gtilde}, we have
        \begin{multline}
        \widehat{K}_{t_1,a_1,b_1,\lambda}(h,h') = - \oint_{\Sigma_1} z^{-h}V(z)z^{h'}\frac{dz}{2\pi iz} \\ 
        -\oint_{\Sigma_1} \oint_{\Sigma_1} z^{-h}\frac{w}{z-w} V(z)F_1(z)^T\Upsilon_\pm(z)^{-1}\Upsilon_\pm(w)G_1(w)^Tw^{h'} \frac{dw}{2\pi iw} \frac{dz}{2\pi iz} \\
         + \oint_{\Sigma_2} \oint_{\Sigma_1} z^{-h}\frac{w}{z-w}F_2(z)^T\Upsilon_\pm(z)^{-1}\Upsilon_\pm(w)G_1(w)^Tw^{h'} \frac{dw}{2\pi iw} \frac{dz}{2\pi iz},
    \end{multline}
    where $F_1$, $F_2$, $G_1$ and $G_2$ are given by \eqref{def: F1}--\eqref{def: G2}, and where $\Upsilon$ is the unique solution of RH problem \ref{prbl: IIKS}. We observe that the function $F_1^T\Upsilon_\pm^{-1}$, given on $\Sigma_1$, has an analytic continuation in the whole of $\mathcal{U}$, given by $F_1^T \Upsilon^{-1}$. Indeed, $F_1^T \Upsilon^{-1}$ is analytic in $\mathcal{U} \setminus \left(\Sigma_1 \cup \Sigma_2 \right)$, and along both $\Sigma_1$ and $\Sigma_2$, $F_1(z)^T \Upsilon_-(z)^{-1} = F_1(z)^T \Upsilon_+(z)^{-1}$. Therefore, in the second term, we can deform the $z$-integration contour from $\Sigma_1$ to $\Sigma_2$ and obtain
        \begin{multline}
        \widehat{K}_{t_1,a_1,b_1,\lambda}(h,h') = - \oint_{\Sigma_1} z^{-h}V(z)z^{h'}\frac{dz}{2\pi iz} \\
        + \oint_{\Sigma_2} \oint_{\Sigma_1} z^{-h}\frac{w}{z-w}\left(F_2(z)^T-V(z)F_1(z)^T\right)\Upsilon_\pm(z)^{-1}\Upsilon_\pm(w)G_1(w)^Tw^{h'} \frac{dw}{2\pi iw} \frac{dz}{2\pi iz},
    \end{multline}
    which yields the result. 
\end{proof}

We now assume that $K$ is of the form \eqref{def: kernelRHP}, and we perform a certain redressing procedure in order to incorporate the two RH problems \ref{prbl: IIKS} and \ref{RHPYmodel}.

\begin{lemma}
We have
\begin{equation}\label{eq:RHidentity}
     \Upsilon(z) \begin{bmatrix}
            Y(z) & \vec 0_{d \times 2} \\
            \vec 0_{2 \times d} & I_2
        \end{bmatrix} = Y(z;t_1,a_1,b_1,\lambda),
\end{equation}
where $\Upsilon$, $Y$ and $Y(\cdot;t_1,a_1,b_1,\lambda)$ are the respective unique solutions of RH problems \ref{prbl: IIKS}, \ref{RHPYmodel} and \ref{prbl: main}.
\end{lemma}
\begin{proof}
The proof follows from showing that both sides of \eqref{eq:RHidentity} solve RH problem \ref{prbl: main}. We start by showing that both sides satisfy the same jump conditions on $\Sigma_1$ and $\Sigma_2$. We have that, for $z\in\Sigma_1$,
\begin{align}
    &\begin{bmatrix}
            Y_-(z)^{-1} & \vec 0_{d \times 2} \\
            \vec 0_{2 \times d} & I_2\end{bmatrix}\Upsilon_-(z)^{-1} \Upsilon_+(z)\begin{bmatrix}
            Y_+(z) & \vec 0_{d \times 2} \\
            \vec 0_{2 \times d} & I_2
        \end{bmatrix} \nonumber \\ 
        & \qquad = \begin{bmatrix}
            Y_-(z)^{-1} & \vec 0_{d \times 2} \\
            \vec 0_{2 \times d} & I_2\end{bmatrix} \begin{bmatrix}
                I_d & \vec \sigma_{t_1}(z)W_{t_1}(z) \vec f_1(z)^T \\
                \vec 0_{2 \times d} & I_2
            \end{bmatrix}\begin{bmatrix}
            Y_+(z) & \vec 0_{d \times 2} \nonumber \\ 
            \vec 0_{2 \times d} & I_2
        \end{bmatrix} \\
        & \qquad = \begin{bmatrix}
                J^{(1)}(z) & \vec \psi_{t_1}(z)W_{t_1}(z) \vec f_1(z)^T \\
                \vec 0_{2 \times d} & I_2
            \end{bmatrix} = J^{(1)}(z;t_1,a_1,b_1,\lambda),
\end{align}
where, for the second equality, we recall the relation \eqref{def:phipsi}. We similarly check that, for $z \in \Sigma_2$,
\begin{equation}
    \begin{bmatrix}
            Y_-(z)^{-1} & \vec 0_{d \times 2} \\
            \vec 0_{2 \times d} & I_2\end{bmatrix}\Upsilon_-(z)^{-1} \Upsilon_+(z)\begin{bmatrix}
            Y_+(z) & \vec 0_{d \times 2} \\
            \vec 0_{2 \times d} & I_2
        \end{bmatrix} = J^{(2)}(z;\vec t, \vec a, \vec b,\lambda).
\end{equation}
It is easily verified that both sides satisfy the same asymptotic condition; as $z \rightarrow \infty$,
\begin{equation}
    \Upsilon(z) \begin{bmatrix}
            Y(z) & \vec 0_{d \times 2} \\
            \vec 0_{2 \times d} & I_2
        \end{bmatrix} = \left(I_{d + 2} + \mathcal{O}(1/z)\right) \begin{bmatrix}
            Y_{\infty}(z) & \vec 0_{d \times 2} \\
            \vec 0_{2 \times d} & I_2
        \end{bmatrix}.
\end{equation}
Therefore, both sides solve RH problem \ref{prbl: main}, and since the solution is unique the result follows.
\end{proof}

Substituting \eqref{eq:RHidentity} into \eqref{eq:Khatidentity}, we obtain
\begin{multline} 
        \widehat{K}_{t_1,a_1,b_1,\lambda}(h,h') = - \oint_{\Sigma_1} z^{-h}V(z)z^{h'}\frac{dz}{2\pi iz} 
        \\ + \oint_{\Sigma_2} \oint_{\Sigma_1} z^{-h} \frac{w}{z-w}\left(F_2(z)^T-V(z)F_1(z)^T\right)\begin{bmatrix}
            Y(z) & \vec 0_{d \times 2} \\
            \vec 0_{2 \times d} & I_2\end{bmatrix}Y_\pm(z;t_1,a_1,b_1)^{-1} \\
            \times\ Y_\pm(w;t_1,a_1,b_1)\begin{bmatrix}
            Y(w)^{-1} & \vec 0_{d \times 2} \\
            \vec 0_{2 \times d} & I_2
        \end{bmatrix} G_1(w)^Tw^{h'} \frac{dw}{2\pi iw} \frac{dz}{2\pi iz}.
    \end{multline}
Then, by \eqref{def: F1}--\eqref{def: G2} and \eqref{def:phipsi}, we obtain  
\begin{multline}
        \widehat{K}_{t_1,a_1,b_1,\lambda}(h,h') = - \oint_{\Sigma_1} z^{-h}V(z)z^{h'}\frac{dz}{2\pi iz} \\ 
        + \oint_{\Sigma_2} \oint_{\Sigma_1} z^{-h} \frac{w}{z-w}\begin{pmatrix}
        \frac{1}{W_{t}(z)} \vec \phi_{t}(z)^T & -1_{t > t_1} \frac{W_{t_1}(z)}{W_{t}(z)}\vec f_1(z)^T \\
        \frac{1}{W_{t'}(z)} \vec \phi_{t'}(z)^T & -1_{t' > t_1} \frac{W_{t_1}(z)}{W_{t'}(z)}\vec f_1(z)^T \\
        \frac{1}{W_{t_1}(z)} \vec \phi_{t_1}(z)^T & 0
    \end{pmatrix}  Y_\pm(z;t_1,a_1,b_1)^{-1}\\
    \times Y_\pm(w;t_1,a_1,b_1)
    \begin{pmatrix}
        \vec \psi_{t}(w)W_{t}(w) & \vec \psi_{t'}(w)W_{t'}(w) & \vec \psi_{t_1}(w)W_{t_1}(w) \\
        -1_{t_1>t} \vec g_{1}(w) \frac{W_{t}(w)}{W_{t_1}(w)} & -1_{t_1>t'} \vec g_{1}(w) \frac{W_{t'}(w)}{W_{t_1}(w)} & \vec 0_2
    \end{pmatrix}w^{h'} \frac{dw}{2\pi iw} \frac{dz}{2\pi iz}.
    \end{multline}
By \eqref{expr: matrix KB hat} we have that
\begin{equation}
    K_{t_1,a_1,b_1,\lambda}(t,h;t,h') = \widehat{K}_{t_1,a_1,b_1,\lambda}(h,h')_{1,2},
\end{equation}
and one can verify that $\widehat{K}_{t_1,a_1,b_1,\lambda}(h,h')_{1,2}$ is the same as $K_{t_1,a_1,b_1,\lambda}(t,h;t',h')$ given in Theorem \ref{thm: main}. This completes the proof of Theorem \ref{thm: main} in the $q=1$ case, and for $t$, $t'$, and $t_1$ distinct.

\subsection{Extension to non-distinct $t,t',t_1$}

The above expression was derived for $t$, $t'$, and $t_1$ all distinct.
If they are not distinct, $\Lambda_{t,t',t_1}$ consists of only one or two copies of the integers, and
$\Phi$ defined by \eqref{def:Phi} is no longer a bijection from $\ell^2(\mathbb Z;\mathbb C^3)$ to $\ell^2(\Lambda_{t,t',t_1})$. Instead it is a bijection from the subspace of $\ell^2(\mathbb Z;\mathbb C^3)$ with either two or three identical entries
to $\ell^2(\Lambda_{t,t',t_1})$. This has no further consequences in the analysis in Section \ref{subsec:matrix}, except for the fact that the operators $\widehat K$ in \ref{def:hatK}, $\widehat K_{t_1,a_1,b_1,\lambda}$ in \eqref{expr: matrix KB hat}, and $\widehat{1}_{B_{t_1, a_1, b_1}}$ in \eqref{def:1B} are only defined on the relevant subspace of $\ell^2(\mathbb Z;\mathbb C^3)$. Note also that some of the indicator functions in \eqref{def: V} vanish.
This modification has no further consequences in the analysis and computations in Sections \ref{subsec:Fourier}, \ref{subsec:toint}, and \ref{subsec:IIKS}.

\section{Ratios of Fredholm determinants} \label{section: ratio}
In this section we suppose that $K$ belongs to the space $\mathcal{S}_d(\vec W)$ and that it is given in the form \eqref{def: kernelRHP}. We consider $B_{\vec t, \vec a, \vec b}$ given by \eqref{def: gap}, for some $q\in\mathbb N$, and $\lambda \in \C$, such that $1-\lambda 1_{B_{\vec t, \vec a, \vec b}}K$ is invertible as an $\ell^2(\Lambda_N)$-operator. For notational simplicity, we will focus on the proof of \eqref{eq:ratioid1}. The proof of \eqref{eq:ratioid2} requires only straightforward adaptations of the arguments and computations. 

\medskip

Recall the definition \eqref{def: gap+} of the $j$-th raised subset $B_{\vec t, \vec a, \vec b}^{j+}$  associated to $B_{\vec t, \vec a, \vec b}$. Observe that
\begin{align}
    \frac{\det\left(1 - \lambda 1_{B_{\vec t, \vec a, \vec b}^{j+}}K\right)_{\ell^2(\Lambda_N)}}{\det\left(1 - \lambda 1_{B_{\vec t, \vec a, \vec b}}K\right)_{\ell^2(\Lambda_N)}} &= \det\left(1 - \lambda 1_{B_{\vec t, \vec a, \vec b}}K - \lambda 1_{\{(t_j,b_j)\}}K\right)_{\ell^2(\Lambda_N)}\det\left(1 - \lambda 1_{B_{\vec t, \vec a, \vec b}}K\right)_{\ell^2(\Lambda_N)}^{-1} \nonumber \\
    &= \det \left(1 - \lambda 1_{\{(t_j,b_j)\}}K\left(1 - \lambda 1_{B_{\vec t, \vec a, \vec b}}K\right)^{-1}\right)_{\ell^2(\Lambda_N)} \nonumber \\
    &= \det \left(1 - \lambda 1_{\{(t_j,b_j)\}}K_{\vec t, \vec a, \vec b,\lambda }\right)_{\ell^2(\Lambda_N)} \nonumber \\
    &= 1 - \lambda K_{\vec t, \vec a, \vec b, \lambda}(t_j,b_j;t_j,b_j).
\end{align}
Therefore, by Theorem \ref{thm: main},
\begin{multline} \label{expr: ratio1}
    \frac{\det\left(1 - \lambda 1_{B_{\vec t, \vec a, \vec b}^{j+}}K\right)_{\ell^2(\Lambda_N)}}{\det\left(1 - \lambda 1_{B_{\vec t, \vec a, \vec b}}K\right)_{\ell^2(\Lambda_N)}} = 1 - \lambda \oint_{\Sigma_2} \oint_{\Sigma_1} u^{-b_j} \frac{1}{W_{t_j}(u)} \frac{v}{u-v} \vec \phi_{t_j}(u;\vec t, \vec a, \vec b, \lambda)^T Y_+(u;\vec t, \vec a, \vec b, \lambda)^{-1} \\
    \times Y_-(v;\vec t, \vec a, \vec b, \lambda) \vec \psi_{t_j}(v;\vec t, \vec a, \vec b, \lambda) W_{t_j}(v) v^{b_j} \frac{dv}{2\pi iv} \frac{du}{2\pi iu},
\end{multline}
where $Y(\cdot;\vec t, \vec a, \vec b, \lambda)$ solves RH problem \ref{prbl: main}. Note that we make specific choices for the boundary values of $Y$. From \eqref{def: phipsi},
\begin{equation}
    \vec \phi_{t_j}(u;\vec t, \vec a, \vec b, \lambda) := \begin{bmatrix}
        \vec \phi_{t_j}(u) \\ \vec 0_{2j} \\ -1_{t_j>t_{j+1}} \vec f_{j+1}(u) W_{t_{j+1}}(u) \\ \vdots \\ -1_{t_j>t_q} \vec f_q(u) W_{t_q}(u)
    \end{bmatrix}, \quad \vec \psi_{t_j}(v;\vec t, \vec a, \vec b, \lambda) := \begin{bmatrix}
        \vec \psi_{t_j}(v) \\ -1_{t_1>t_j} \vec g_1(v) \frac{1}{W_{t_1}(v)} \\
        \vdots \\
        -1_{t_{j-1}>t_j} \vec g_{j-1}(v) \frac{1}{W_{t_{j-1}}(v)} \\ \vec 0_{2(q-j+1)}
    \end{bmatrix}.
\end{equation}
We introduce new notation. For $k,k'=1, \dots, q$, we define
\begin{align}
    S_k(z) &:= \vec \psi_{t_k}(z)W_{t_k}(z)f_k(z)^T, \\
    T_k(z) &:= g_k(z)\frac{1}{W_{t_k}(z)}\vec \phi_{t_k}(z)^T, \\
    U_{k,k'}(z) &:= 1_{t_k > t_k'} g_k(z) \frac{W_{t_{k'}}(z)}{W_{t_k}(z)}f_{k'}(z)^T.
\end{align}
Using this we can rewrite the jump matrices of RH problem \ref{prbl: main} given by \eqref{def: J1} and \eqref{def: J2} in the following way.
\begin{multline}
    J^{(1)}(z;\vec t, \vec a, \vec b, \lambda) = \\ \left[\begin{array}{c|c|c|c}
        J^{(1)}(z) & \begin{array}{ccc}
            S_1(z) & \cdots & S_{j-1}(z)
        \end{array} & S_j(z) & \begin{array}{ccc}
            S_{j+1}(z) & \cdots & S_{q}(z)
        \end{array} \\ \hline
        \vec 0_{2(j-1) \times d} & \begin{array}{ccc}
             I_2 & \cdots & -U_{1,j-1}(z) \\
             \vdots & \ddots & \vdots \\
             \vec 0_{2\times 2} & \cdots & I_2
         \end{array} & \begin{array}{c}
            -U_{1,j}(z) \\ \vdots \\ -U_{j-1,j}(z)
         \end{array} & \begin{array}{ccc}
             -U_{1,j+1}(z) & \cdots & -U_{1,q}(z) \\
             \vdots & \ddots & \vdots \\
             -U_{j-1,j+1}(z) & \cdots & -U_{j-1,q}(z)
         \end{array} \\ \hline
         \vec 0_{2 \times d} & \vec 0_{2 \times 2(j-1)} & I_2 & \begin{array}{ccc}
            -U_{j,j+1}(z) & \cdots & -U_{j,q}(z)
         \end{array} \\ \hline 
         \vec 0_{2(q-j+1)\times d} & \vec 0_{2(q-j+1)\times 2(j-1)} & \vec 0_{2(q-j+1)\times 2} & \begin{array}{ccc}
             I_2 & \cdots & -U_{j+1,q}(z) \\
             \vdots & \ddots & \vdots \\
             \vec 0_{2\times 2} & \cdots & I_2
         \end{array}
    \end{array}\right],
\end{multline}
and
\begin{multline}
    J^{(2)}(z;\vec t, \vec a, \vec b, \lambda) = \left[\begin{array}{c|c|c|c}
        J^{(2)}(z) & \vec 0_{d \times 2(j-1)} & \vec 0_{d \times 2} & \vec 0_{d \times 2(q-j+1)} \\ \hline
        \begin{array}{c}
            T_1(z) \\ \vdots \\ T_{j-1}(z)
        \end{array} & I_{2(j-1)} & \vec 0_{2(j-1) \times 2} & \vec 0_{2(j-1) \times 2(q-j+1)} \\ \hline 
        T_j(z) & \vec 0_{2 \times 2(j-1)} & I_2 & \vec 0_{2 \times 2(q-j+1)} \\ \hline
        \begin{array}{c}
            T_{j+1}(z) \\ \vdots \\ T_q(z)
        \end{array} & \vec 0_{2(q-j+1) \times 2(j-1)} & \vec 0_{2(q-j+1) \times 2} & I_{2(q-j+1)}
    \end{array}\right].
\end{multline}
We now introduce a transformation of RH problem \ref{prbl: main} by defining
\begin{equation}\label{def:Ytilde}
    \tilde{Y}(z;\vec t, \vec a, \vec b, \lambda) := \begin{cases}
        Y(z;\vec t, \vec a, \vec b, \lambda), \quad & \text{inside $\Sigma_1$,} \\
        Y(z;\vec t, \vec a, \vec b, \lambda) \tilde{J}(z), \quad & \text{between $\Sigma_1$ and $\Sigma_2$,} \\
        Y(z;\vec t, \vec a, \vec b, \lambda), \quad & \text{outside $\Sigma_2$,}
    \end{cases}
\end{equation}
where
\begin{equation}
    \tilde{J}(z) :=
    \left[\begin{array}{c|c|c}
        I_{d + 2(j-1)} & \vec 0_{(d+2(j-1)) \times 2} & \vec 0_{(d+2(j-1)) \times 2(q-j+1)} \\ \hline
        \vec 0_{2 \times (d+2(j-1))} & I_2 & \begin{array}{ccc}
            -U_{j,j+1}(z) & \cdots & -U_{j,q}(z)
         \end{array} \\ \hline 
         \vec 0_{2(q-j+1)\times (d+2(j-1))} & \vec 0_{2(q-j+1)\times 2} & \begin{array}{ccc}
             I_2 & \cdots & -U_{j+1,q}(z) \\
             \vdots & \ddots & \vdots \\
             \vec 0_{2\times 2} & \cdots & I_2
         \end{array}
    \end{array}\right].
\end{equation}
The effect of this transformation is that a multiplicative factor $\tilde J$ of the jump matrix $J^{(2)}$ on $\Sigma_2$ is transferred to the jump matrix $J^{(1)}$ on $\Sigma_1$. More precisely, $\tilde Y$ will have modified jump matrices
\begin{equation}\label{def:J12tilde}
\tilde J^{(2)}=J^{(2)}\tilde J,\qquad \tilde J^{(1)}=\tilde J^{-1}J^{(1)}.
\end{equation}
It follows that $\tilde{Y}(z;\vec t, \vec a, \vec b, \lambda)$ solves the following RH problem:
\begin{problem}
    \begin{enumerate}
            \item $\tilde Y(\cdot;\vec t, \vec a, \vec b, \lambda): \mathbb{C} \setminus (\Sigma_1 \cup \Sigma_2) \rightarrow \mathbb{C}^{(d+2q) \times (d+2q)}$ is analytic.
            \item 
            For $z \in \Sigma_1 \cup \Sigma_2$, $\tilde Y(.;\vec t, \vec a, \vec b, \lambda)$ has continuous boundary values, $\tilde Y_+(z;\vec t, \vec a, \vec b, \lambda)$ and $\tilde Y_-(z;\vec t, \vec a, \vec b, \lambda)$, and they satisfy the jump relations
            \begin{align}
            \tilde Y_+(z;\vec t, \vec a, \vec b, \lambda)&=\tilde Y_-(z;\vec t, \vec a, \vec b, \lambda) \tilde J^{(1)}(z;\vec t, \vec a, \vec b, \lambda), \qquad & \text{for $z\in\Sigma_1$},\\
            \tilde Y_+(z;\vec t, \vec a, \vec b, \lambda)&=\tilde Y_-(z;\vec t, \vec a, \vec b, \lambda)\tilde J^{(2)}(z;\vec t, \vec a, \vec b, \lambda), \qquad & \text{for $z\in\Sigma_2$}.
            \end{align}
            \item $\tilde Y(z;\vec t, \vec a, \vec b, \lambda) = \left(I_{d+2q} + \mathcal{O}(1/z)\right) \begin{bmatrix}
                Y^{(\infty)}(z) & \vec0_{d \times 2q} \\
                \vec 0_{2q \times d} & I_{2q}
            \end{bmatrix}$ as $z \rightarrow \infty$.
        \end{enumerate}
\end{problem}
The new jump matrices are given explicitly by
\begin{multline}\label{def:tildeJ1}
    \tilde J^{(1)}(z;\vec t, \vec a, \vec b, \lambda) = \\ \left[\begin{array}{c|c|c|c}
        J^{(1)}(z) & \begin{array}{ccc}
            S_1(z) & \cdots & S_{j-1}(z)
        \end{array} & S_j(z) & \begin{array}{ccc}
            S_{j+1}(z) & \cdots & S_{q}(z)
        \end{array} \\ \hline
        \vec 0_{2(j-1) \times d} & \begin{array}{ccc}
             I_2 & \cdots & -U_{1,j-1}(z) \\
             \vdots & \ddots & \vdots \\
             \vec 0_{2\times 2} & \cdots & I_2
         \end{array} & \begin{array}{c}
            -U_{1,j}(z) \\ \vdots \\ -U_{j-1,j}(z)
         \end{array} & \begin{array}{ccc}
             -U_{1,j+1}(z) & \cdots & -U_{1,q}(z) \\
             \vdots & \ddots & \vdots \\
             -U_{j-1,j+1}(z) & \cdots & -U_{j-1,q}(z)
         \end{array} \\ \hline
         \vec 0_{2 \times d} & \vec 0_{2 \times 2(j-1)} & I_2 & \vec 0_{2 \times 2(q-j+1)} \\ \hline 
         \vec 0_{2(q-j+1)\times d} & \vec 0_{2(q-j+1)\times 2(j-1)} & \vec 0_{2(q-j+1)\times 2} & I_{2(q-j+1)}
    \end{array}\right],
\end{multline}
and
\begin{multline}\label{def:tildeJ2}
    \tilde J^{(2)}(z;\vec t, \vec a, \vec b, \lambda) = \left[\begin{array}{c|c|c|c}
        J^{(2)}(z) & \vec 0_{d \times 2(j-1)} & \vec 0_{d \times 2} & \vec 0_{d \times 2(q-j+1)} \\ \hline
        \begin{array}{c}
            T_1(z) \\ \vdots \\ T_{j-1}(z)
        \end{array} & I_{2(j-1)} & \vec 0_{2(j-1) \times 2} & \vec 0_{2(j-1) \times 2(q-j+1)} \\ \hline 
        T_j(z) & \vec 0_{2 \times 2(j-1)} & I_2 & \begin{array}{ccc}
            -U_{j,j+1}(z) & \cdots & -U_{j,q}(z)
         \end{array} \\ \hline
        \begin{array}{c}
            T_{j+1}(z) \\ \vdots \\ T_q(z)
        \end{array} & \vec 0_{2(q-j+1) \times 2(j-1)} & \vec 0_{2(q-j+1) \times 2} & \begin{array}{ccc}
             I_2 & \cdots & -U_{j+1,q}(z) \\
             \vdots & \ddots & \vdots \\
             \vec 0_{2\times 2} & \cdots & I_2
         \end{array}
    \end{array}\right].
\end{multline}
One verifies that
\begin{equation}
    \vec \phi_{t_j}(z;\vec t, \vec a, \vec b, \lambda)^T \tilde J(z) = \vec \phi_{t_j}(z;\vec t, \vec a, \vec b, \lambda)^T, \quad \tilde J(z)^{-1} \vec \psi_{t_j}(z;\vec t, \vec a, \vec b, \lambda) = \vec \psi_{t_j}(z;\vec t, \vec a, \vec b, \lambda),
\end{equation}
such that we can substitute $Y(\cdot;\vec t, \vec a, \vec b)$ for $\tilde Y(\cdot;\vec t, \vec a, \vec b)$ in \eqref{expr: ratio1} to obtain
\begin{multline}\label{eq:ratioidproof1}
    \frac{\det\left(1 - \lambda 1_{B_{\vec t, \vec a, \vec b}^{j+}}K\right)_{\ell^2(\Lambda_N)}}{\det\left(1 - \lambda 1_{B_{\vec t, \vec a, \vec b}}K\right)_{\ell^2(\Lambda_N)}} = 1 - \lambda \oint_{\Sigma_2} \oint_{\Sigma_1} u^{-b_j} \frac{1}{W_{t_j}(u)} \frac{v}{u-v} \vec \phi_{t_j}(u;\vec t, \vec a, \vec b)^T \tilde Y_+(u;\vec t, \vec a, \vec b)^{-1} \\
    \times \tilde Y_-(v;\vec t, \vec a, \vec b) \vec \psi_{t_j}(v;\vec t, \vec a, \vec b) W_{t_j}(v) v^{b_j} \frac{dv}{2\pi iv} \frac{du}{2\pi iu}.
\end{multline}

\medskip

Let $e_{d+2j-1}$ denote the standard unit vector of size $d+2q$ with a $1$ in the $(d+2j-1)$-position. We observe that, for $z \in \Sigma_1$,
\begin{align}
    \tilde Y_+(z;\vec t, \vec a, \vec b, \lambda) e_{d+2j-1} &= \tilde Y_-(z;\vec t, \vec a, \vec b, \lambda) \tilde J^{(1)}(z;\vec t, \vec a, \vec b, \lambda)e_{d+2j-1} \nonumber \\
    &= \tilde Y_-(z;\vec t, \vec a, \vec b, \lambda) \left(\lambda \vec \psi_{t_j}(z;\vec t, \vec a, \vec b, \lambda)W_{t_j}(z) z^{b_j} + e_{d+2j-1}\right).
\end{align}
This implies that
\begin{multline}
    \oint_{\Sigma_2} \oint_{\Sigma_1} u^{-b_j} \frac{1}{W_{t_j}(u)} \frac{v}{u-v} \vec \phi_{t_j}(u;\vec t, \vec a, \vec b, \lambda)^T \tilde Y_\pm(u;\vec t, \vec a, \vec b, \lambda)^{-1} \\
    \times \tilde Y_-(v;\vec t, \vec a, \vec b, \lambda) \left(\lambda \vec \psi_{t_j}(v;\vec t, \vec a, \vec b, \lambda)W_{t_j}(v) v^{b_j} + e_{d+2j-1}\right) \frac{dv}{2\pi iv} \frac{du}{2\pi iu} \\ = \oint_{\Sigma_2} \oint_{\Sigma_1} u^{-b_j} \frac{1}{W_{t_j}(u)} \frac{v}{u-v} \vec \phi_{t_j}(u;\vec t, \vec a, \vec b, \lambda)^T \tilde Y_\pm(u;\vec t, \vec a, \vec b, \lambda)^{-1} \\ \times \tilde Y_+(v;\vec t, \vec a, \vec b, \lambda) e_{d+2j-1} \frac{dv}{2\pi iv} \frac{du}{2\pi iu} = 0,
\end{multline}
since the integrand extends to an analytic function of $v$ inside $\Sigma_1$. We therefore obtain from \eqref{eq:ratioidproof1} that
\begin{multline} \label{expr: H1}
\frac{\det\left(1 - \lambda 1_{B_{\vec t, \vec a, \vec b}^{j+}}K\right)_{\ell^2(\Lambda_N)}}{\det\left(1 - \lambda 1_{B_{\vec t, \vec a, \vec b}}K\right)_{\ell^2(\Lambda_N)}} = 1 + \oint_{\Sigma_2} \oint_{\Sigma_1} u^{-b_j} \frac{1}{W_{t_j}(u)} \frac{v}{u-v} \vec \phi_{t_j}(u;\vec t, \vec a, \vec b, \lambda)^T \tilde Y_\pm(u;\vec t, \vec a, \vec b, \lambda)^{-1} \\ \times \tilde Y_-(v;\vec t, \vec a, \vec b, \lambda) e_{d+2j-1} \frac{dv}{2\pi iv} \frac{du}{2\pi iu}. 
\end{multline}
Similarly, for $z \in \Sigma_2$,
\begin{align*}
    e_{d+2j-1}^T \tilde Y_-(z;\vec t, \vec a, \vec b, \lambda)^{-1} &= e_{d+2j-1}^T J^{(2)}(z;\vec t, \vec a, \vec b, \lambda) \tilde Y_+(z;\vec t, \vec a, \vec b, \lambda)^{-1}, \nonumber \\
    &= \left(z^{-b_j} \frac{1}{W_{t_j}(z)} \vec \phi_{t_j}(z;\vec t, \vec a, \vec b, \lambda)^T + e_{d+2j-1}^T\right) \tilde Y_+(z;\vec t, \vec a, \vec b, \lambda)^{-1},
\end{align*}
such that
\begin{multline*}
    \oint_{\Sigma_2} \oint_{\Sigma_1} \frac{v}{u-v} \left(u^{-b_j} \frac{1}{W_{t_j}(u)} \vec \phi_{t_j}(u;\vec t, \vec a, \vec b, \lambda)^T + e_{d+2j-1}^T\right) \tilde Y_+(u;\vec t, \vec a, \vec b, \lambda)^{-1} \\ \times \tilde Y_-(v;\vec t, \vec a, \vec b, \lambda) e_{d+2j-1} \frac{dv}{2\pi iv} \frac{du}{2\pi iu} \\ = \oint_{\Sigma_2} \oint_{\Sigma_1} \frac{v}{u-v} e_{d+2j-1}^T \tilde Y_-(u;\vec t, \vec a, \vec b, \lambda)^{-1} \\ \times \tilde Y_-(v;\vec t, \vec a, \vec b, \lambda) e_{d+2j-1} \frac{dv}{2\pi iv} \frac{du}{2\pi iu} = 0,
\end{multline*}
since the integrand extends to an analytic function of $u$ outside $\Sigma_2$, and
\[\text{$\frac{v}{u-v}e_{d+2j-1}^T \tilde Y(u;\vec t, \vec a, \vec b, \lambda)^{-1} = \mathcal{O}(1/u)$ as $u \rightarrow \infty$.}\] Hence, by \eqref{expr: H1}, we obtain
\begin{multline} \label{expr: H2}
    \frac{\det\left(1 - \lambda 1_{B_{\vec t, \vec a, \vec b}^{j+}}K\right)_{\ell^2(\Lambda_N)}}{\det\left(1 - \lambda 1_{B_{\vec t, \vec a, \vec b}}K\right)_{\ell^2(\Lambda_N)}} = 1
      - \oint_{\Sigma_2} \oint_{\Sigma_1} \frac{v}{u-v} e_{d+2j-1}^T \tilde Y_+(u;\vec t, \vec a, \vec b, \lambda)^{-1} \\ \times\ \tilde Y_-(v;\vec t, \vec a, \vec b, \lambda) e_{d+2j-1} \frac{dv}{2\pi iv} \frac{du}{2\pi iu}. 
\end{multline}
Next, we observe that
\begin{align*} 
    &\oint_{\Sigma_2} \oint_{\Sigma_1} \frac{v}{u-v}  e_{d+2j-1}^T \tilde Y_+(u;\vec t, \vec a, \vec b, \lambda)^{-1} \tilde Y_-(v;\vec t, \vec a, \vec b, \lambda) e_{d+2j-1} \frac{dv}{2\pi iv} \frac{du}{2\pi iu} \\
    &= 1 + \oint_{\Sigma_1} \oint_{\Sigma_2} \frac{v}{u-v} e_{d+2j-1}^T \tilde Y_-(u;\vec t, \vec a, \vec b, \lambda)^{-1} \tilde Y_+(v;\vec t, \vec a, \vec b, \lambda) e_{d+2j-1} \frac{dv}{2\pi iv} \frac{du}{2\pi iu}.\end{align*} 
    To see this, we deform the $u$-integration contour from $\Sigma_2$ to $\Sigma_1$ and vice-versa for the $v$-integration contour. In doing so we pick up the residue at $u=v$, which yields the above identity.
  By \eqref{expr: H2}, we obtain   
    \begin{align*}
&\frac{\det\left(1 - \lambda 1_{B_{\vec t, \vec a, \vec b}^{j+}}K\right)_{\ell^2(\Lambda_N)}}{\det\left(1 - \lambda 1_{B_{\vec t, \vec a, \vec b}}K\right)_{\ell^2(\Lambda_N)}} \\
&\ = 
- \oint_{\Sigma_1} \oint_{\Sigma_2} \frac{v}{u-v} e_{d+2j-1}^T \tilde Y_-(u;\vec t, \vec a, \vec b, \lambda)^{-1} \tilde Y_+(v;\vec t, \vec a, \vec b, \lambda) e_{d+2j-1} \frac{dv}{2\pi iv} \frac{du}{2\pi iu}
\\&\  = - \oint_{\Sigma_1} \oint_{\Sigma_2} \frac{v}{u-v} e_{d+2j-1}^T \tilde Y_+(u;\vec t, \vec a, \vec b, \lambda)^{-1} \tilde Y_-(v;\vec t, \vec a, \vec b, \lambda) e_{d+2j-1} \frac{dv}{2\pi iv} \frac{du}{2\pi iu}, \end{align*}
where we used the jump relations $\tilde Y_+=\tilde Y_- \tilde J^{(1)}$ on $\Sigma_1$ and $\tilde Y_+=\tilde Y_-\tilde J^{(2)}$ on $\Sigma_2$ as well as the expressions \eqref{def:tildeJ1}--\eqref{def:tildeJ2}.
Then, since the $u$-integrand is analytic inside $\Sigma_1$ except for a simple pole at $u=0$, we compute the $u$-integral using the residue theorem and obtain
\[\frac{\det\left(1 - \lambda 1_{B_{\vec t, \vec a, \vec b}^{j+}}K\right)_{\ell^2(\Lambda_N)}}{\det\left(1 - \lambda 1_{B_{\vec t, \vec a, \vec b}}K\right)_{\ell^2(\Lambda_N)}}=
       \oint_{\Sigma_2} e_{d+2j-1}^T Y(0;\vec t, \vec a, \vec b, \lambda)^{-1} \tilde Y_-(v;\vec t, \vec a, \vec b, \lambda) e_{d+2j-1} \frac{dv}{2\pi iv}.\] 
Finally, since the $v$-integrand is analytic outside $\Sigma_2$, we can evaluate the $v$-integral by computing the residue at infinity. Recalling that $\lim_{z\to\infty}\tilde Y(z;\vec t, \vec a, \vec b, \lambda)=I_{d+2q}$, we obtain 
       \[\frac{\det\left(1 - \lambda 1_{B_{\vec t, \vec a, \vec b}^{j+}}K\right)_{\ell^2(\Lambda_N)}}{\det\left(1 - \lambda 1_{B_{\vec t, \vec a, \vec b}}K\right)_{\ell^2(\Lambda_N)}}= e_{d+2j-1}^T \tilde Y(0;\vec t, \vec a, \vec b, \lambda)^{-1} e_{d+2j-1}.\]
By \eqref{def:Ytilde}, the right hand side is equal to the $(d+2j-1,d+2j-1)$-entry of the matrix $Y(0;\vec t, \vec a, \vec b, \lambda)^{-1}$. Equivalently, it is equal to the
$(d+2j-1,d+2j-1)$-entry of $Y(0;\vec t, \vec a, \vec b, \lambda)^{-T}$, and this completes the proof of \eqref{eq:ratioid1}. We prefer this 
expression in terms of the inverse transpose of $Y$ because $Y^{-T}$ itself solves a convenient RH problem, with right-multiplicative jump matrices, recall Remark \ref{remark:inversetranspose}.

\section{Domino tilings of reduced  Aztec diamonds}\label{section:4}

\paragraph{Domino tilings of the Aztec diamond.} The Aztec diamond $A_N$ of size $N$ is the union of all the squares $[m,m+1]\times[n,n+1] \subset \mathbb{R}^2$, $m,n\in \mathbb{Z}$, which lie within the region $\{(x,y) \in \mathbb{R}^2: |x|+|y| \leq N+1 \}$. A horizontal domino is a $1 \times 2$ rectangle given by $[m,m+2]\times[n,n+1] \subset \mathbb{R}^2$, for $m,n\in \mathbb{Z}$ and a vertical domino is a $2 \times 1$ rectangle given by $[m,m+1]\times[n,n+2] \subset \mathbb{R}^2$, for $m,n\in \mathbb{Z}$. A domino tiling of $A_N$ is a configuration of horizontal and vertical dominoes that cover the whole domain $A_N$ with disjoint interiors, see Figure \ref{fig:AD} (left). The set of all tilings of $A_N$ is denoted by $\mathcal{T}(A_N)$.
By colouring the Aztec diamond in a chequerboard fashion so that the leftmost square of each row in the top half is white, we are able to distinguish four types of dominoes. We call a horizontal domino a north/south domino if its leftmost square is white/black, and we call a vertical domino a west/east domino if its upper square is white/black.

\begin{figure}[t!]

\begin{subfigure}{0.5\textwidth}
\begin{center}
\includegraphics[keepaspectratio, width=0.8\linewidth]{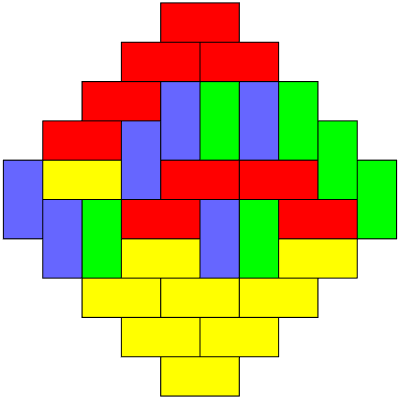} 
\end{center}
\end{subfigure}
\begin{subfigure}{0.5\textwidth}
\begin{center}
\includegraphics[keepaspectratio, width=0.8\linewidth]{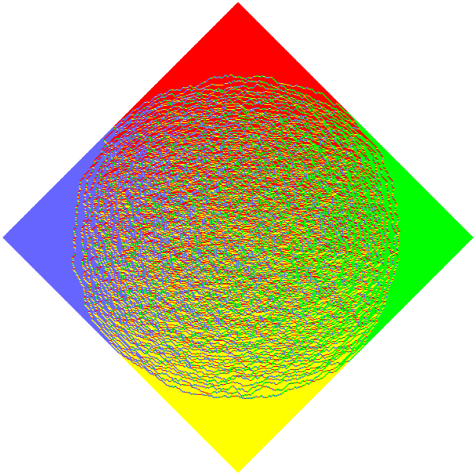} 
\end{center}
\end{subfigure}

\caption{Left: A domino tiling of $A_5$. the north, south, east, and west dominoes are shown in red, yellow, green, and blue respectively. Right: A domino tiling of $A_{300}$. Both tilings are sampled uniformly.}
\label{fig:AD}
\end{figure}

\medskip

Domino tilings of Aztec diamonds were introduced in \cite{elkies1991alternatingsignmatricesdomino, elkies1991alternatingsignmatricesdomino2} where the number of tilings of $A_N$ was calculated to be $2^{N(N+1)/2}$. More generally, it was shown that, for $a \in (0,1),$
\begin{equation}
    F_N(a) := \sum_{T \in \mathcal{T}(A_N)} a^{v(T)} = (1+a^2)^{N(N+1)/2},
\end{equation}
where $v(T)$ is the number of vertical dominoes in a tiling $T$. 
We consider the probability distribution on $\mathcal{T}(A_N)$ given by 
\begin{equation}\label{def:probtiling}
\mathbb P_N\left(T\right):=\frac{a^{v(T)}}{(1+a^2)^{N(N+1)/2}}.
\end{equation}
Here, the uniform distribution corresponds to $a=1$. One observes that under $\mathbb P_N$, a random domino tiling of $A_N$ exhibits for large $N$, with overwhelming probability, frozen regions near the corners of the Aztec diamond where only one type of domino is present, and a rough region where all types of dominoes occur, see Figure \ref{fig:AD} (right). The existence of the asymptotic boundary separating the rough and frozen regions, called the arctic curve, was first established in \cite{jockusch1998randomdominotilingsarctic}; see also \cite{cohn1996localstatistics, cohn2001variationalprincipal}. 
Local correlations between dominoes are described by the Airy line ensemble near the arctic curve \cite{Johansson2003TheAC}, except near the points where the arctic curve meets the boundary of the Aztec diamond, where they are described by the GUE-minors/corners process \cite{Johansson2006GUEminors}.
In the rough region, local correlations and height function fluctuations are connected to Gibbs measures \cite{Chhita2015asymptoticdominostatistics, Kenyon2006amoebae}
and the Gaussian Free Field \cite{Kenyon2000DominosAT, Bufetov2016FluctuationsOP}.

\medskip

Random domino tilings of the Aztec diamond can be translated to point processes. Under the probability measure \eqref{def:probtiling}, this point process has the remarkable feature of being determinantal, with an explicit double-contour integral expression for its correlation kernel \cite{Johansson2003TheAC}.
There are only a few other determinantal domino tiling models in which explicit expressions for the correlation kernel are known. One example consists of domino tilings of the Aztec diamond with probability measures induced by doubly-periodic weights, whose asymptotic behaviour has been studied extensively in recent years \cite{Chhita2016dominostatistics, Chhita2014couplingfunctions, Chhita2015asymptoticdominostatistics, Duits2017TheTA, Berggren2019toeplitz, Berggren2021, Berggren2025geometry, Borodin2023biased}. A key result for this model is an exact expression for the correlation kernel established in \cite{Duits2017TheTA}, which has strong similarities with \eqref{def:kernelresult}. 
Other examples consist of domino tilings on domains consisting of two Aztec diamonds glued together in specific ways, see \cite{ADLER2014518, adler2015, adler2018nonconvexpolygons}.
Here, we will consider domino tilings of reduced Aztec diamonds, introduced and studied in \cite{CP2013, CP2015, CPS2019, KP2016, charlier2025countingdominolozengetilings, CharlierClaeys2025asymptotics}.
While this model is known to be determinantal, more precisely it is a multi-time version of a discrete orthogonal polynomial ensemble, no explicit expression for its multi-time correlation has been obtained so far. We will apply Theorem \ref{thm: main} in order to obtain an exact double-contour integral expression for the multi-time correlation kernel.

\paragraph{Non-intersecting paths.} Before stating our result, we briefly summarise how one obtains the point process underlying domino tilings of the Aztec diamond, arising from an equivalent representation as configurations of non-intersecting paths, as shown in Figure \ref{fig:ADpaths}. We closely follow the explanation given in \cite{Johansson2003TheAC} where more details can be found.

\begin{figure}[t!]

\begin{subfigure}{0.5\textwidth}
\begin{center}
\includegraphics[keepaspectratio, width=0.8\linewidth]{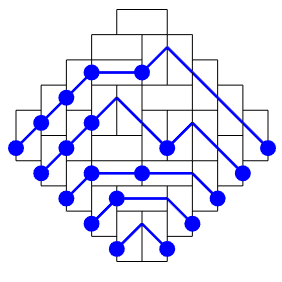} 
\end{center}
\end{subfigure}
\begin{subfigure}{0.5\textwidth}
\begin{center}
\includegraphics[keepaspectratio, width=0.8\linewidth]{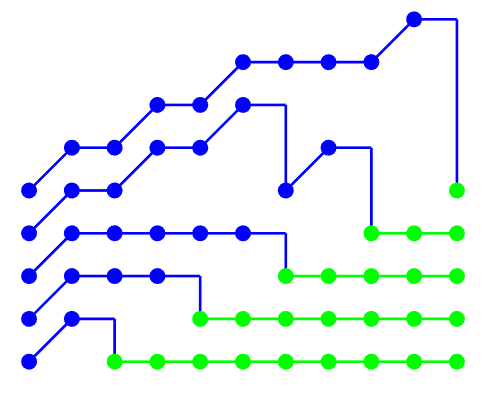} 
\end{center}
\end{subfigure}

\caption{For the same tiling of $A_5$, the paths are shown on the left and the corresponding non-intersecting paths are shown on the right.}
\label{fig:ADpaths}
\end{figure}

\medskip 

We mark a domino tiling of $A_N$ by drawing an ascending diagonal line segment on each west domino, a descending diagonal line segment on each east domino, and a horizontal line segment on each south domino. This results in a collection of $N$ non-intersecting paths, as shown in Figure \ref{fig:ADpaths} (left).

\medskip

We transform this collection of non-intersecting paths into a new collection of non-intersecting paths on the lattice $\Lambda_{2N}$ as follows. Starting with the top path, we draw a new path starting at the point $(0,0)$. 
Reading the original path from left to right, we translate every up-right step (corresponding to a west domino) to the combination of an up-right step $(t,h)\to (t+1,h+1)$ followed by a right step $(t+1,h+1)\to (t+2,h+1)$; we translate every right step (corresponding to a south domino) to a double right step $(t,h)\to (t+1,h)\to (t+2,h)$; and we translate every down-right step (corresponding to an east domino) to a single down step $(t,h)\mapsto (t,h-1)$.
This procedure is repeated for the other paths marking the domino tiling, but now by starting at lattice points $(0,-1),(0,-2),\ldots, (0,-N+1)$. 
This results in $N$ non-intersecting lattice paths with starting points $(0,0),(0,-1),\ldots, (0,-N+1)$ and endpoints $(2N,0), (2N-2,-1),\ldots, (2,-N+1)$, as shown in Figure \ref{fig:ADpaths} (right).
We complement these paths (which depend on the domino tiling) with horizontal lines (which are independent of the chosen domino tiling) connecting $(2N,0), (2N-2,-1),\ldots, (2,-N+1)$ with $(2N,0), (2N,-1),\ldots, (2N,-N+1)$, and with an infinite number of horizontal paths connecting $(0,-N), (0,-N-1), (0,-N-2),\ldots $ with $(2N,-N), (2N,-N-1), (2N,-N-2),\ldots $.

\medskip

We draw a point at each lattice point of $\Lambda_{2N}$ where a path passes, except when it precedes a down step, again see Figure \ref{fig:ADpaths} (right). From this collection of points, one can first reconstruct the non-intersecting paths on the lattice, secondly the non-intersecting paths marking the domino tiling, and finally the domino tiling itself.
In other words, there is a bijection between domino tilings of $A_N$ and such admissible point configurations. Here, admissible means that the points correspond to an infinite configuration of non-intersecting paths
starting at $(0,-1),(0,-2),\ldots$ and ending at $(2N,0), (2N,-1),\ldots$, which take up-right steps $(2s,h)\to (2s+1,h+1)$, right steps $(2s,h)\to (2s+1,h)$ or down steps $(2s,h), (2s,h-1)$ at even time levels $t=2s$, and right steps $(2s-1,h)\to (2s,h)$ at odd time levels $t=2s-1$, and which can only take right steps $(t,h)\to (t+1,h)$ for $2h\leq -2N+t$.

\medskip

Since every vertical domino corresponds to an up-right or a down step of a path, the probability distribution on $\mathcal T(A_N)$ can be transported to a probability distribution on the set of admissible point configurations, by assigning probability 
$\frac{a^{v(T)}}{(1+a^2)^{N(N+1)/2}}$
to each admissible point configuration $T$, where $v(T)$ is now the total number of up-right and down steps in the corresponding configuration of paths.

\medskip
 
The above construction may seem artificial, but it is very powerful because the resulting random configuration of points $\mathcal X$ forms a determinantal point process which we call $\mathbb{P}_{N}$.
The correlation kernel for this determinantal point process is given by  \cite{Johansson2003TheAC}
\begin{multline}\label{def:kernelAD}
    K_N(t,h;t',h') = -1_{t>t'} \oint_{\Sigma_1} z^{-h} \frac{W_{t'}(z)}{W_t(z)} z^{h'} \frac{dz}{2\pi iz} \\
    + \oint_{\Sigma_2} \oint_{\Sigma_1} u^{-h} \frac{1}{W_t(u)} \frac{v}{u-v} W_{t'}(v) v^{h'} \frac{dv}{2\pi iv} \frac{du}{2\pi iu}, 
\end{multline}
where $\Sigma_1$ is a closed contour encircling the points $0$ and $a$ but not $-1/a$, $\Sigma_2$ is a closed contour encircling $\Sigma_1$, and for $t = 2s-\epsilon$, for $s \in \{1, \dots, N\}$ and $\epsilon \in \{0,1\}$,
\begin{equation} \label{def: ADW}
    W_t(z) = \frac{z^{N-s+\epsilon - 1}}{(z-a)^{N-s+\epsilon}(1+az)^{s}}, \quad z \in \mathbb{C}.
\end{equation}
For $0<a<1$, we can take $\Sigma_1$ to be the unit circle and then $K_N$ is precisely of the form \eqref{def: K1} with $d=1$ and $\vec\rho=\vec\sigma=1$. Observe therefore that $K_N$ belongs to the space $\mathcal{S}_1(\vec W)$ with the specific sequence of functions $W_t$ given by \eqref{def: ADW}.

\paragraph{Domino tilings of reduced Aztec diamonds.} A reduced Aztec diamond is the Aztec diamond with an approximate rectangle removed, see Figure \ref{fig:ADreduced}. Such regions were first introduced by Colomo and Pronko in \cite{colomo2015}. More precisely, for $s = 1, \dots, N$ and $k = 1, \dots , s+1$, the reduced Aztec diamond $A_N^{s,k}$ is the union of all the lattice squares within the region
\[\left\{(x,y) \in \mathbb{R}^2 : |x|+|y| \leq N+1 \,\, \text{and} \,\, |x+N-2s+k| \geq y - k + 1  \right\}.\]
This results in $N-s+1$ diagonal rows of $s-k+1$ horizontal dominoes being removed from the Aztec diamond. The set of all tilings of $A_N^{s,k}$ is denoted by $\mathcal{T}(A_N^{s,k})$. Alternatively, we can view tilings of $A_N^{s,k}$ as tilings of $A_N$ for which the removed region $A_N \setminus A_N^{s,k}$ is tiled with a frozen region of north dominoes. Such tilings of $A_N$ correspond in turn to admissible point configurations $\mathcal{X}$ which contain no points in the subset
\begin{equation} \label{def: gapAD}
    B_N^{s,k} := \{2s-1\} \times \{k,k+1,\dots,s\} \subset \Lambda_{2N}.
\end{equation}
Therefore, in the form \eqref{def: pregap}, $B_N^{s,k} = B_{2s-1,k,s+1}$.

\begin{figure}[t!]

\begin{subfigure}{0.5\textwidth}
\begin{center}
\includegraphics[keepaspectratio, width=0.8\linewidth]{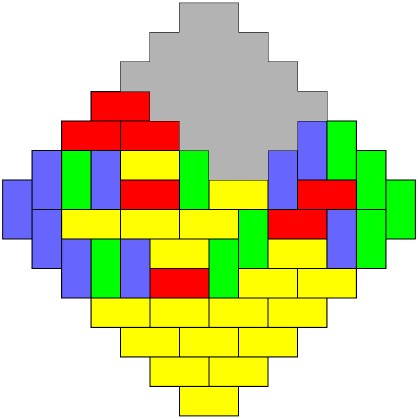} 
\end{center}
\end{subfigure}
\begin{subfigure}{0.5\textwidth}
\begin{center}
\includegraphics[keepaspectratio, width=0.8\linewidth]{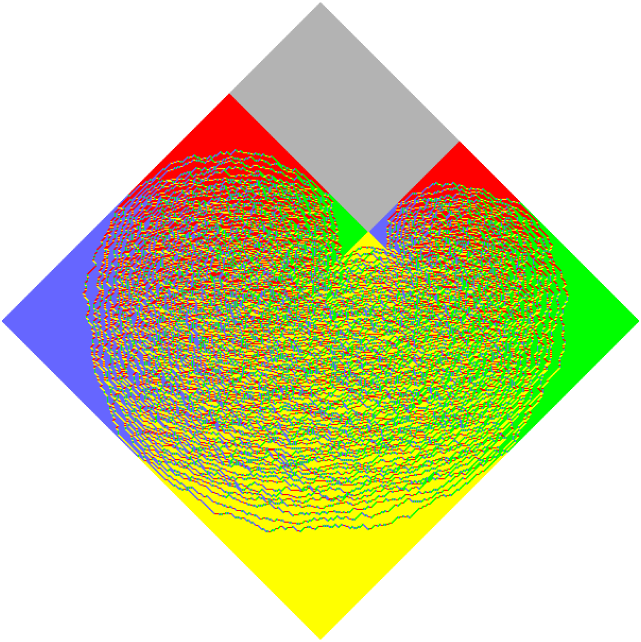} 
\end{center}
\end{subfigure}

\caption{Left: A domino tiling of $A_7^{2,5}$. Right: A domino tiling of $A_{300}^{215,85}$. Both tilings are sampled uniformly.}
\label{fig:ADreduced}
\end{figure}

\medskip

We define a probability distribution on the tilings of $A_N^{s,k}$ with the same system of weights as for tilings of $A_N$, i.e. vertical dominoes have weight $a$ and horizontal dominoes have weight $1$, such that
\begin{equation}
    \mathbb P_N^{s,k}\left(T\right)=\frac{a^{v(T)}}{\sum_{T'\in\mathcal T(A_N^{s,k})}a^{v(T')}}.
\end{equation}
Thus the distribution of tilings of $A_N^{s,k}$ is given by the distribution of tilings of $A_N$ conditioned on the event that $A_N \setminus A_N^{s,k}$ is frozen with north dominoes. Equivalently, on the level of point configurations $\mathcal X$, we condition $\mathbb{P}_N$ on the event that $\mathcal{X} \cap B_N^{s,k} = \emptyset$. Therefore, this produces a new point process $\mathbb{P}_N^{s,k}$ on $\Lambda_{2N} \setminus B_N^{s,k}$ whose correlation kernel is given by the integral kernel of the operator
\begin{equation}\label{def:ADkernel}
    K_N^{s,k} := \left(K_N\right)_{2s-1,k,s+1}.
\end{equation}
We can therefore apply Theorem \ref{thm: main} to obtain an expression of the form \eqref{def:kernelresult}
for the kernel $K_N^{s,k}$, with $Y_\pm$ expressed in terms of a $3\times 3$ RH problem. This RH problem can be further simplified and reduced to the following RH problem of size $2\times 2$.
\begin{problem} \label{prbl: AD}
    \begin{enumerate}
            \item $Y: \mathbb{C} \setminus (\Sigma_1 \cup \Sigma_2) \rightarrow \mathbb{C}^{2 \times 2}$ is analytic.
            \item 
            For $z \in \Sigma_1 \cup \Sigma_2$, $Y$ has continuous boundary values, $Y_+$ and $Y_-$, and they satisfy the jump relations
            \begin{align}
            Y_+(z)&=Y_-(z)\begin{pmatrix}
                1 & -z^kW_{2s-1}(z) \\
                0 & 1
            \end{pmatrix}, \qquad & \text{for $z\in\Sigma_1$},\\
            Y_+(z) &= Y_-(z)\begin{pmatrix}
                1 & 0 \\ 
                \frac{1}{z^{k}W_{2s-1}(z)} & 1
            \end{pmatrix}, \qquad & \text{for $z\in\Sigma_2$}.
            \end{align}
            \item $Y(z) = I_{2} + \mathcal{O}(1/z)$ as $z \rightarrow \infty$.
        \end{enumerate}
\end{problem}

\begin{theorem}
    Let $s\in\{1, \dots, N\}$ and $k \in\{1, \dots, s+1\}$. The kernel $K_N^{s,k}$ of the determinantal point process $\mathbb P_N^{s,k}$ is given by
    \begin{multline}\label{eq:kernelADresult}
        K_N^{s,k}(t,h;t',h') = (1_{t>2s-1}1_{2s-1>t'}-1_{t>t'}) \oint_{\Sigma_1} z^{-h} \frac{W_{t'}(z)}{W_t(z)} z^{h'} \frac{dz}{2\pi iz} \\
        + \oint_{\Sigma_2} \oint_{\Sigma_1} u^{-h} \frac{1}{W_{t}(u)} \frac{v}{u-v} \vec \phi_t(u)^T Y_+(u)^{-1}
        Y_-(v) \vec \psi_{t'}(v) W_{t'}(v) v^{h'} \frac{dv}{2\pi iv} \frac{du}{2\pi iu},
    \end{multline}
        where
    \begin{equation}
        \vec \phi_t(u) = \begin{pmatrix}
            1 \\ 1_{t>2s-1}u^{k}W_{2s-1}(u)
        \end{pmatrix}, \quad \vec \psi_{t'}(v) = \begin{pmatrix}
            1 \\ -1_{2s-1>t'}\frac{1}{v^{k}W_{2s-1}(v)}
        \end{pmatrix},
    \end{equation}
    and where $Y$ solves RH problem \ref{prbl: AD}.
    
\end{theorem}
\begin{remark}
We emphasize that the choices of boundary values $Y_+,Y_-$ in 
\eqref{eq:kernelADresult} is important, in contrast to the choices in \eqref{def:kernelresult}. This will become clear in our proof below.
\end{remark}

\begin{proof}
    By applying Theorem \ref{thm: main} to compute $\left(K_N\right)_{2s-1,k,s+1}$, we obtain
    \begin{multline} \label{4.8}
        K_N^{s,k}(t,h;t',h') = -1_{t>t'} \oint_{\Sigma_1} z^{-h} \frac{W_{t'}(z)}{W_t(z)} z^{h'} \frac{dz}{2\pi iz} \\
        + \oint_{\Sigma_2} \oint_{\Sigma_1} u^{-h} \frac{1}{W_{t}(u)} \frac{v}{u-v} \begin{pmatrix}
            1 & -1_{t>2s-1}u^{s+1}W_{2s-1}(u) & 1_{t>2s-1}u^{k}W_{2s-1}(u)
        \end{pmatrix} U_+(u)^{-1} \\
        \times U_-(v) \begin{pmatrix}
            1 \\ -1_{2s-1>t'}\frac{1}{v^{s+1}W_{2s-1}(v)} \\ -1_{2s-1>t'}\frac{1}{v^{k}W_{2s-1}(v)}
        \end{pmatrix} W_{t'}(v) v^{h'} \frac{dv}{2\pi iv} \frac{du}{2\pi iu},
    \end{multline}
    where $U$ solves the following RH problem.
    \begin{problem}
    \begin{enumerate}
            \item $U: \mathbb{C} \setminus (\Sigma_1 \cup \Sigma_2) \rightarrow \mathbb{C}^{3 \times 3}$ is analytic.
            \item 
            For $z \in \Sigma_1 \cup \Sigma_2$, $U$ has continuous boundary values, $U_+$ and $U_-$, and they satisfy the jump relations
            \begin{align}
            U_+(z)&=U_-(z)\begin{pmatrix}
                1 & z^{s+1}W_{2s-1}(z) & -z^kW_{2s-1}(z) \\
                0 & 1 & 0 \\
                0 & 0 & 1
            \end{pmatrix}, \qquad & \text{for $z\in\Sigma_1$},\\
            U_+(z) &= U_-(z)\begin{pmatrix}
                1 & 0 & 0 \\
                \frac{1}{z^{s+1}W_{2s-1}(z)} & 1 & 0\\ 
                \frac{1}{z^{k}W_{2s-1}(z)} & 0 & 1
            \end{pmatrix}, \qquad & \text{for $z\in\Sigma_2$}.
            \end{align}
            \item $U(z) = I_{3} + \mathcal{O}(1/z)$ as $z \rightarrow \infty$.
        \end{enumerate}
        \end{problem}
        We proceed by showing how we can reduce this RH problem of size $3\times 3$ to one of size $2\times 2$. We start with a transformation defined as follows,
        \begin{equation} \label{4.9}
    V(z) := \begin{cases}
        \begin{pmatrix}
            1 & 0 & 0 \\
            0 & 0 & 1 \\
            -a^s & 1 & 0 \\
        \end{pmatrix} U(z) \begin{pmatrix}
            1 & 0 & 0 \\
            \frac{1}{z^{s+1}W_{2s-1}(z)} & 0 & 1 \\
            0 & 1 & 0
        \end{pmatrix}, \quad & \text{for $z$ outside $\Sigma_2$,} \\
        \begin{pmatrix}
            1 & 0 & 0 \\
            0 & 0 & 1 \\
            -a^s & 1 & 0 \\
        \end{pmatrix} U(z) \begin{pmatrix}
            1 & 0 & 0 \\
            0 & 0 & 1 \\
            0 & 1 & 0
        \end{pmatrix}, \quad & \text{for $z$ inside $\Sigma_2$.}
    \end{cases}
    \end{equation}
    Using the fact that
    \begin{equation*}
        \frac{1}{z^{s+1}W_{2s-1}(z)} = \frac{(z-a)^{N-s+1}(1+az)^s}{z^{N+1}} = a^s + \mathcal{O}(1/z), \quad \text{as $z \rightarrow \infty$,}
    \end{equation*}
    we see that $V(z) = I_{3} + \mathcal{O}(1/z)$ as $z \rightarrow \infty$ and thus verify that $V$ solves the following RH problem.
    \begin{problem}
    \begin{enumerate}
            \item $V: \mathbb{C} \setminus (\Sigma_1 \cup \Sigma_2) \rightarrow \mathbb{C}^{3 \times 3}$ is analytic.
            \item 
            For $z \in \Sigma_1 \cup \Sigma_2$, $V$ has continuous boundary values, $V_+$ and $V_-$, and they satisfy the jump relations
            \begin{align}
            V_+(z)&=V_-(z)\begin{pmatrix}
                1 & -z^kW_{2s-1}(z) & z^{s+1}W_{2s-1}(z) \\
                0 & 1 & 0 \\
                0 & 0 & 1
            \end{pmatrix}, \qquad & \text{for $z\in\Sigma_1$},\\
            V_+(z) &= V_-(z)\begin{pmatrix}
                1 & 0 & 0 \\
                \frac{1}{z^{k}W_{2s-1}(z)} & 1 & 0\\ 
                0 & 0 & 1
            \end{pmatrix}, \qquad & \text{for $z\in\Sigma_2$}.
            \end{align}
            \item $V(z) = I_{3} + \mathcal{O}(1/z)$ as $z \rightarrow \infty$.
        \end{enumerate}
        \end{problem}
 Furthermore, by substituting \eqref{4.9} into \eqref{4.8}, we obtain
    \begin{multline} \label{4.10}
        K_N^{s,k}(t,h;t',h') = -1_{t>t'} \oint_{\Sigma_1} z^{-h} \frac{W_{t'}(z)}{W_t(z)} z^{h'} \frac{dz}{2\pi iz}\\
        + \oint_{\Sigma_2} \oint_{\Sigma_1} u^{-h} \frac{1}{W_{t}(u)} \frac{v}{u-v} \begin{pmatrix}
            1 & 1_{t>2s-1}u^{k}W_{2s-1}(u) & -1_{t>2s-1}u^{s+1}W_{2s-1}(u)
        \end{pmatrix} V_+(u)^{-1} \\
        \times V_-(v) \begin{pmatrix}
            1 \\ -1_{2s-1>t'}\frac{1}{v^{k}W_{2s-1}(v)} \\ -1_{2s-1>t'}\frac{1}{v^{s+1}W_{2s-1}(v)}
        \end{pmatrix} W_{t'}(v) v^{h'} \frac{dv}{2\pi iv} \frac{du}{2\pi iu}.
    \end{multline}
    We introduce
    \begin{multline}
        I := \oint_{\Sigma_2} \oint_{\Sigma_1} u^{-h} \frac{1}{W_{t}(u)} \frac{v}{u-v} \begin{pmatrix}
            1 & 1_{t>2s-1}u^{k}W_{2s-1}(u) & -1_{t>2s-1}u^{s+1}W_{2s-1}(u)
        \end{pmatrix} V_+(u)^{-1} \\
        \times V_-(v) \begin{pmatrix}
            0 \\ 0 \\ -1_{2s-1>t'}\frac{1}{v^{s+1}W_{2s-1}(v)}
        \end{pmatrix} W_{t'}(v) v^{h'} \frac{dv}{2\pi iv} \frac{du}{2\pi iu},
    \end{multline}
    and evaluate it by deforming the $v$-integration contour $\Sigma_1$ to a large circle $\Sigma_1'$ outside $\Sigma_2$. To do so, we observe that $$V(v) \begin{pmatrix}
            0 \\ 0 \\ -1_{2s-1>t'}\frac{1}{v^{s+1}W_{2s-1}(v)}
        \end{pmatrix} W_{t'}(v) v^{h'}$$ is analytic outside $\Sigma_1$. This follows from both the jump condition satisfied by $V$ on $\Sigma_2$, and the fact that $$1_{2s-1>t'}\frac{v^{h'}W_{t'}(v)}{v^{s+1}W_{2s-1}(v)}$$ is analytic outside $\Sigma_1$ since the singularity at $-1/a$ is removable. 
        Therefore, in deforming $\Sigma_1$ to $\Sigma_1'$, we only pick up the residue at $v=u$, and we obtain
    \begin{multline}
        I =1_{t>2s-1}1_{2s-1>t'} \oint_{\Sigma_1} z^{-h} \frac{W_{t'}(z)}{W_t(z)} z^{h'} \frac{dz}{2\pi iz} \\ +\oint_{\Sigma_2} \oint_{\Sigma_1'} u^{-h} \frac{1}{W_{t}(u)} \frac{v}{u-v} \begin{pmatrix}
            1 & 1_{t>2s-1}u^{k}W_{2s-1}(u) & -1_{t>2s-1}u^{s+1}W_{2s-1}(u)
        \end{pmatrix} V_+(u)^{-1} \\
        \times V(v) \begin{pmatrix}
            0 \\ 0 \\ -1_{2s-1>t'}\frac{1}{v^{s+1}W_{2s-1}(v)}
        \end{pmatrix} W_{t'}(v) v^{h'} \frac{dv}{2\pi iv} \frac{du}{2\pi iu}.
    \end{multline}
    However, since
    \[\frac{v^{h'}W_{t'}(v)}{v^{s+1}W_{2s-1}(v)} = \mathcal{O}(1/v), \quad \text{as $v \rightarrow \infty$,}\]
    the double integral converges to $0$ as we deform $\Sigma_1'$ towards infinity, and thus
    \begin{equation} \label{4.12}
        I = 1_{t>2s-1}1_{2s-1>t'} \oint_{\Sigma_1} z^{-h} \frac{W_{t'}(z)}{W_t(z)} z^{h'} \frac{dz}{2\pi iz}.
    \end{equation}
    Substituting \eqref{4.12} back into \eqref{4.10}, we obtain
    \begin{multline} \label{4.13}
        K_N^{s,k}(t,h;t',h') = \left(1_{t>2s-1}1_{2s-1>t'}-1_{t>t'}\right) \oint_{\Sigma_1} z^{-h} \frac{W_{t'}(z)}{W_t(z)} z^{h'} \frac{dz}{2\pi iz} \\
        + \oint_{\Sigma_2} \oint_{\Sigma_1} u^{-h} \frac{1}{W_{t}(u)} \frac{v}{u-v} \begin{pmatrix}
            1 & 1_{t>2s-1}u^{k}W_{2s-1}(u) & -1_{t>2s-1}u^{s+1}W_{2s-1}(u)
        \end{pmatrix} V_+(u)^{-1} \\
        \times V_-(v) \begin{pmatrix}
            1 \\ -1_{2s-1>t'}\frac{1}{v^{k}W_{2s-1}(v)} \\ 0
        \end{pmatrix} W_{t'}(v) v^{h'} \frac{dv}{2\pi iv} \frac{du}{2\pi iu}.
    \end{multline}

    \medskip 

    The inverse-transpose of $V$ solves the following RH problem.
    \begin{problem}
    \begin{enumerate}
            \item $V^{-T}: \mathbb{C} \setminus (\Sigma_1 \cup \Sigma_2) \rightarrow \mathbb{C}^{3 \times 3}$ is analytic.
            \item 
            For $z \in \Sigma_1 \cup \Sigma_2$, $V^{-T}$ has continuous boundary values, $V^{-T}_+$ and $V^{-T}_-$, and they satisfy the jump relations
            \begin{align}
            V^{-T}_+(z)&=V^{-T}_-(z)\begin{pmatrix}
                1 & 0 & 0 \\
                z^kW_{2s-1}(z) & 1 & 0 \\
                -z^{s+1}W_{2s-1}(z) & 0 & 1
            \end{pmatrix}, \qquad & \text{for $z\in\Sigma_1$},\\
            V^{-T}_+(z) &= V^{-T}_-(z)\begin{pmatrix}
                1 & -\frac{1}{z^{k}W_{2s-1}(z)} & 0 \\
                0 & 1 & 0\\ 
                0 & 0 & 1
            \end{pmatrix}, \qquad & \text{for $z\in\Sigma_2$}.
            \end{align}
            \item $V^{-T}(z) = I_{3} + \mathcal{O}(1/z)$ as $z \rightarrow \infty$.
        \end{enumerate}
        \end{problem}
    From this we conclude that the third column of $V^{-T}$, denoted by $\left(V^{-T}\right)_3$, is analytic, and furthermore,
    $$\left(V^{-T}\right)_3(z) = \begin{pmatrix}
        0 \\ 0 \\ 1
    \end{pmatrix} \quad \text{as $z \rightarrow \infty$.}$$
    Therefore by Liouville's theorem, $$\left(V^{-T}\right)_3 \equiv \begin{pmatrix}
        0 \\ 0 \\ 1
    \end{pmatrix}.$$
    We can therefore write $V^{-T}$ in the form
    \[V^{-T}(z)=\begin{pmatrix}A(z)&\vec 0_{2\times 1}\\ \vec w(z)&1\end{pmatrix},\]
    for some invertible $2\times 2$ matrix $A(z)$ and some $1\times 2$ vector $\vec w(z)$. Taking the inverse transpose of this expression, we obtain
    \begin{equation}
        V(z) = \begin{bmatrix}
            Y(z) & -Y(z)\vec w^T(z) \\
            \vec 0_{1 \times 2} & 1
        \end{bmatrix},\qquad Y=A^{-T}.\ 
    \end{equation}
    Moreover, substituting this expression in the RH conditions for $V$, we find that $Y$ solves RH problem \ref{prbl: AD}. 
        We compute
        \begin{align}
            &\begin{pmatrix}
            1 & 1_{t>2s-1}u^{k}W_{2s-1}(u) & -1_{t>2s-1}u^{s+1}W_{2s-1}(u)
        \end{pmatrix} V_+(u)^{-1}V_-(v) \begin{pmatrix}
            1 \\ -1_{2s-1>t'}\frac{1}{v^{k}W_{2s-1}(v)} \\ 0
        \end{pmatrix} \nonumber \\
        &\quad= \begin{pmatrix}
            1 & 1_{t>2s-1}u^{k}W_{2s-1}(u) & -1_{t>2s-1}u^{s+1}W_{2s-1}(u)
        \end{pmatrix} \begin{bmatrix}
                Y_+(u)^{-1} & \vec w_+(u)^T \\
                \vec 0_{1 \times 2} & 1
            \end{bmatrix} \nonumber \\
            &\qquad\times \begin{bmatrix}
            Y_-(v) & -Y_-(v)\vec w_-(v)^T \\
            \vec 0_{1 \times 2} & 1
        \end{bmatrix} \begin{pmatrix}
            1 \\ -1_{2s-1>t'}\frac{1}{v^{k}W_{2s-1}(v)} \\ 0
        \end{pmatrix} \nonumber \\
        &\quad = \begin{pmatrix}
            1 & 1_{t>2s-1}u^{k}W_{2s-1}(u)
        \end{pmatrix} Y_+(u)^{-1}Y_-(v) \begin{pmatrix}
            1 \\ -1_{2s-1>t'}\frac{1}{v^{k}W_{2s-1}(v)}
        \end{pmatrix}.
        \end{align}
Substituting this back into \eqref{4.13}, we complete the proof.
\end{proof}

\begin{remark}
 We could apply Theorem  \ref{thm: ratio} together with the above reduction of the RH problem to characterise the ratios of Fredholm determinants as follows,
    \begin{equation}\frac{\det \left(1 - 1_{B_N^{s,k+1}}K\right)}{\det \left(1 - 1_{B_N^{s,k}}K\right)} = Y^{-T}(0)_{2,2},\end{equation}
    in terms of the solution of RH problem \ref{prbl: AD}.
    This result is a rewriting of \cite[Proposition 1.9]{charlier2025countingdominolozengetilings}, so we do not go further into this.
\end{remark}

\begin{remark}
    One can produce more complicated reduced Aztec diamonds by conditioning on voids of points in multiple clusters of the form \eqref{def: gapAD}. This would result in multiple approximate rectangles being removed from the domain. Repeatedly iterating the above argument would again allow us to express the correlation kernel as a double contour integral. The associated RH problem would have size 1 more than the number of rectangles removed from the domain. 
\end{remark}

\section{Lozenge tilings of reduced hexagons}\label{section:5}
\textbf{Lozenge tilings of hexagons.} Given positive integers $L,M$, and $N$ with $M<L$, we consider the hexagon $H_{L,M,N}$ with corners at $(0,0), (L-M,0), (L,M), (L,M+N), (M,M+N)$, and $(0,N)$. We consider tilings of this domain with three types of lozenges (\begin{tikzpicture}[scale=0.2] \draw (0,0) -- (1,0) -- (1,1) -- (0,1) -- (0,0); \end{tikzpicture}, \begin{tikzpicture}[scale=0.2] \draw (0,0) -- (1,1) -- (1,2) -- (0,1) -- (0,0); \end{tikzpicture}, \begin{tikzpicture}[scale=0.2] \draw (0,0) -- (1,0) -- (2,1) -- (1,1) -- (0,0); \end{tikzpicture}), see Figure \ref{fig:hexagon} (left). The number of tilings of $H_{L,M,N}$, denoted by $G_{L,M,N}$ is given by MacMahon's formula \cite{MacMahon}
\begin{equation}
    G_{L,M,N} = \prod_{i=1}^{L-M} \prod_{j=1}^M \prod_{k=1}^N \frac{i+j+k-1}{i+j+k-2}.
\end{equation}

\begin{figure}[t!]

\begin{subfigure}{0.5\textwidth}
\begin{center}
\includegraphics[keepaspectratio, width=0.8\linewidth]{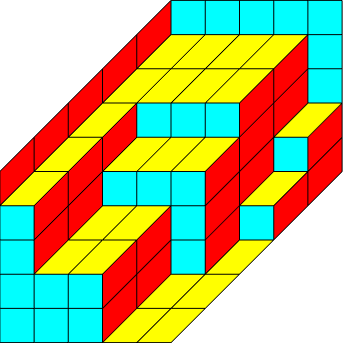} 
\end{center}
\end{subfigure}
\begin{subfigure}{0.5\textwidth}
\begin{center}
\includegraphics[keepaspectratio, width=0.8\linewidth]{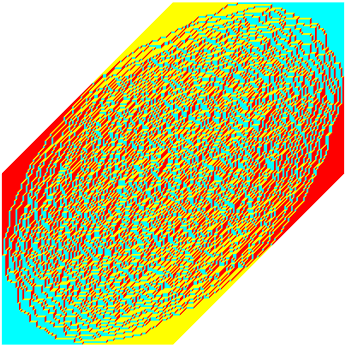} 
\end{center}
\end{subfigure}

\caption{Left: A lozenge tiling of $H_{10,5,5}$. Right: A lozenge tiling of $H_{200,100,100}$. Both tilings are sampled uniformly.}
\label{fig:hexagon}
\end{figure}

\medskip

There is a vast literature on lozenge tilings of hexagons and other planar domains; see \cite{Gorin} for an overview. As for domino tilings of the Aztec diamond, a uniformly random lozenge tiling of a large hexagon exhibits frozen regions in the corners of the hexagon and a liquid region in the centre; the arctic curve is given by the inscribed ellipse \cite{Cohn1998}, see Figure \ref{fig:hexagon} (right).
Local correlations around the arctic curve are described by the Airy line ensemble \cite{petrov2012} or by the GUE-minors/corners process near the points of tangency between the curve and the hexagon \cite{Johansson2006GUEminors, AggarwalGorin2022}.
Local correlations and height function fluctuations in the rough region are described by Gibbs measures \cite{baik2007, gorin2008, borodin2010, petrov2012, Aggarwal2023} and the Gaussian free field \cite{Kenyon2008, Kenyon2007, Petrov2015}.

\medskip

Similarly to domino tilings of the Aztec diamond, lozenge tilings of hexagons can be represented as configurations of points on a lattice, which under the uniform probability distribution on the set of tilings, form a determinantal point process \cite{Johansson2000NonintersectingPR, Duits2017TheTA}. In this section we consider lozenge tilings of reduced hexagons, introduced in \cite{charlier2025countingdominolozengetilings}, and find an exact expression for the correlation kernel of the corresponding determinantal point process. For exact formulas for correlation kernels associated to lozenge tilings of other modified domains, see \cite{petrov2012, adler2018nonconvexpolygons}. We note that the number of lozenges of each type is the same for every tiling of a hexagon, such that assigning different weights to different types of lozenges has no effect and does not lead to biased probability measures. 
One can however consider more complicated probability distributions related to doubly-periodic weights as in \cite{charlier2020periodichexagon, arno2025}.

\paragraph{Non-intersecting paths.} Given a lozenge tiling of the hexagon $H_{L,M,N}$, we construct a collection of $N$ non-intersecting paths on $H_{L,M,N}$  by drawing the following edges on the individual lozenges, \begin{tikzpicture}[scale=0.2] \draw (0,0) -- (1,0) -- (1,1) -- (0,1) -- (0,0) (0,.5) -- (1,.5); \draw [fill] (0,.5) circle [radius=.1]; \draw [fill] (1,.5) circle [radius=.1]; \end{tikzpicture}, \begin{tikzpicture}[scale=0.2] \draw (0,0) -- (1,1) -- (1,2) -- (0,1) -- (0,0) (0,.5) -- (1,1.5); \draw [fill] (0,.5) circle [radius=.1]; \draw [fill] (1,1.5) circle [radius=.1]; \end{tikzpicture}, \begin{tikzpicture}[scale=0.2] \draw (0,0) -- (1,0) -- (2,1) -- (1,1) -- (0,0); \end{tikzpicture}. After shifting down a half step, this produces $N$ non-intersecting paths on the lattice $\Lambda_L$ which start at $0, \dots, N-1$ and end at $M, \dots, M+N-1$. The paths either make horizontal steps to the right or diagonal up-right steps, see Figure \ref{fig:hexagon paths}.

\begin{figure}[t!]

\begin{subfigure}{0.5\textwidth}
\begin{center}
\includegraphics[keepaspectratio, width=0.8\linewidth]{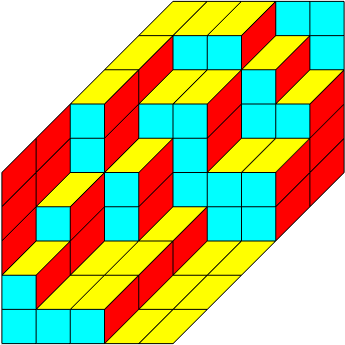} 
\end{center}
\end{subfigure}
\begin{subfigure}{0.5\textwidth}
\begin{center}
\includegraphics[keepaspectratio, width=0.8\linewidth]{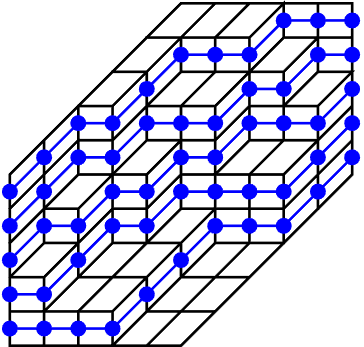} 
\end{center}
\end{subfigure}

\caption{For the same tiling of $H_{10,5,5,}$ shown on the left, the corresponding non-intersecting paths are shown on the right.}
\label{fig:hexagon paths}
\end{figure}

\medskip

The points marking the position of each path at each time level form a determinantal point process $\mathbb{P}_{L,M,N}$ on $\Lambda_L$. An explicit expression for the correlation kernel was obtained in \cite{Duits2017TheTA}, 
\begin{multline} \label{def: kernelhexagon}
    K_{L,M,N}(t,h;t',h') = -1_{t>t'} \oint_{\Sigma_1} z^{-h} \frac{W_{t'}(z)}{W_t(z)} z^{h'} \frac{dz}{2\pi iz} \\
    + \oint_{\Sigma_1} \oint_{\Sigma_1} u^{-h} \frac{1}{W_t(u)} \frac{v}{u-v} \vec \phi(u)^T Y_\pm(u)^{-1} Y_\pm \vec \psi(v) W_{t'}(v) v^{h'}  \frac{dv}{2\pi iv} \frac{du}{2\pi iu},
\end{multline}
where $\Sigma_1$ is a closed, positively oriented contour encircling $0$. For $t = 0, \dots, L$,
\begin{equation} \label{def: HW}
    W_{t}(z) = (z+1)^{L-t},
\end{equation}
and,
\begin{equation}
    \vec \phi(u) = \begin{pmatrix}
        (u+1)^L \\ 0
    \end{pmatrix}, \quad \vec \psi(v) = \begin{pmatrix}
        0 \\ v^{-(M+N)}
    \end{pmatrix}.
\end{equation}
Furthermore, $Y$ solves the following RH problem:
\begin{problem} \label{prbl: hexagon}
    \begin{enumerate}
    \item $Y : \C \setminus \Sigma_1 \rightarrow \C^{2 \times 2}$ is analytic,
    \item For $z \in \Sigma_1$,
    \begin{equation}
        Y_+(z) = Y_-(z) \begin{pmatrix}
            1 & 0 \\
            -\frac{(z+1)^L}{z^{M+N}} & 1
        \end{pmatrix}
    \end{equation}
    \item $Y(z) = \big(I+ O(1/z)\big) \begin{pmatrix}
        z^{-N} & 0 \\ 0 & z^{N}
    \end{pmatrix}$, as $z \rightarrow \infty$.
\end{enumerate}
\end{problem}
By analytic continuation of the integrands in \eqref{def: kernelhexagon}, we can freely deform the integration contours as long as we let them go around $0$. In order to have a kernel of our required form \eqref{def: kernelRHP}, we let the $z$- and $v$-integration contours coincide with the unit circle $\Sigma_1$, and we choose a positively oriented loop $\Sigma_2$ outside the unit circle as the $u$-integration contour. Likewise, we artificially introduce a trivial jump along $\Sigma_2$ into RH problem \ref{prbl: hexagon}. In doing so, $K_{L,M,N}$ becomes of the form \eqref{def: kernelRHP} and thus belongs to the space $S_2(\vec W)$, with the specific sequence of analytic functions given by \eqref{def: HW}. Observe that the parametrisation \eqref{def: kernelRHP} with $Y$ different from the identity matrix is very natural here.
We should note that there exist also other determinantal expressions and formulas for the correlation kernel of lozenge tilings of hexagons, see e.g.\ \cite{Johansson2000NonintersectingPR, petrov2012}, but these do not seem to fit directly in our framework.

\begin{figure}[t!]

\begin{subfigure}{0.5\textwidth}
\begin{center}
\includegraphics[keepaspectratio, width=0.8\linewidth]{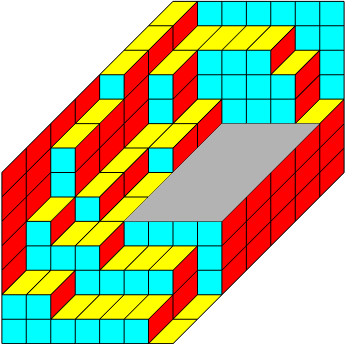} 
\end{center}
\end{subfigure}
\begin{subfigure}{0.5\textwidth}
\begin{center}
\includegraphics[keepaspectratio, width=0.8\linewidth]{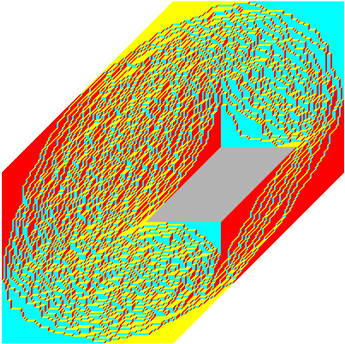} 
\end{center}
\end{subfigure}

\caption{Left: A lozenge tiling of $H_{14,7,7}^{9,5,9}$. Right: A lozenge tiling of $H_{200,100,100}^{128,71,115}$. Both tilings are sampled uniformly.}
\label{fig:reducedhexagon}
\end{figure}

\paragraph{Lozenge tilings of reduced hexagons} We consider reduced hexagons from which a parallelogram is removed, see figure \ref{fig:reducedhexagon}. For $r \in \{1, \dots, L-1\}$ and $\max\{0,r+M-L\} \leq a < b \leq \min\{N+r,N+M\}$, $P_{L,M,N}^{r,a,b}$ denotes the the parallelogram with corners $(r,a)$, $(r-a+b,b)$, $(r,b)$, and $(r-b+a,a)$. We define the reduced hexagon $$H_{L,M,N}^{r,a,b} := H_{L,M,N} \setminus P_{L,M,N}^{r,a,b}.$$ A lozenge tiling of $H_{L,M,N}^{r,a,b}$ can be viewed as a lozenge tiling of $H_{L,M,N}$ for which $P_{L,M,N}^{r,a,b}$ is tiled with a frozen region of \begin{tikzpicture} [scale=0.2] \draw (0,0) -- (1,0) -- (2,1) -- (1,1) -- (0,0); \end{tikzpicture} lozenges. Such tilings of $H_{L,M,N}$ correspond in turn to admissible point configurations $\mathcal{X}$ which contain no points in the subset
\begin{equation} \label{def: gapH}
    B_{L,M,N}^{r,a,b} := \{r\} \times \{a, a+1, \dots, b-1\}.
\end{equation}
We equip the set of lozenge tilings of $H_{L,M,N}^{r,a,b}$ with the uniform probability distribution. This is equivalent to the uniformly distribution of random lozenge tilings of $H_{L,M,N}$ conditioned on the event that $P_{L,M,N}^{r,a,b}$ is frozen with \begin{tikzpicture} [scale=0.2] \draw (0,0) -- (1,0) -- (2,1) -- (1,1) -- (0,0); \end{tikzpicture} lozenges. Therefore, in terms of the associated point process, we condition $\mathbb{P}_{L,M,N}$ on the event that $\mathcal{X} \cap B_{L,M,N}^{r,a,b} = \emptyset$. This produces a new point process $\mathbb{P}_{L,M,N}^{r,a,b}$ on $\Lambda_L \setminus B_{L,M,N}^{r,a,b}$ whose correlation kernel is given by the integral kernel of the operator
\begin{equation}
    K_{L,M,N}^{r,a,b} := \left(K_{L,M,N}\right)_{r,a,b}.
\end{equation}
Thus we can apply Theorem \ref{thm: main} to obtain an expression of the form \eqref{def:kernelresult}
for the kernel $K_{L,M,N}^{r,a,b}$. The kernel is expressed implicitly in terms of the solution of RH problem \ref{prbl: main} which in this particular case becomes the following $4 \times 4$ RH problem.
\begin{problem} \label{prbl: hexagon2}
    \begin{enumerate}
    \item $Y(\cdot;r,a,b) : \C \setminus \Sigma_1 \rightarrow \C^{4 \times 4}$ is analytic,
    \item For $z \in \Sigma_1\cup\Sigma_2$, $Y(\cdot;r,a,b)$ has continuous boundary values, $Y_+(\cdot;r,a,b)$ and $Y_-(\cdot;r,a,b)$, and they satisfy the jump relations
    \begin{align}
            Y_+(z;r,a,b)&=Y_-(z;r,a,b)J^{(1)}(z;r,a,b), \qquad & \text{for $z\in\Sigma_1$},\\
            Y_+(z;r,a,b)&=Y_-(z;r,a,b)J^{(2)}(z;r,a,b), \qquad & \text{for $z\in\Sigma_2$},
            \end{align}
    with
    \begin{align}
     \label{J1hexagon}   J^{(1)}(z;r,a,b) &= \begin{pmatrix}
            1 & 0 & 0 & 0\\
            -\frac{(z+1)^L}{z^{M+N}} & 1 & z^{b-M-N}(z+1)^{L-r} & -z^{a-M-N}(z+1)^{L-r} \\
            0 & 0 & 1 & 0 \\
            0 & 0 & 0 & 1
        \end{pmatrix}, \\
        \label{J2hexagon}
        J^{(2)}(z;r,a,b) &= \begin{pmatrix}
            1 & 0 & 0 & 0\\
            0 & 1 & 0 & 0 \\
            z^{-b}(z+1)^r & 0 & 1 & 0 \\
            z^{-a}(z+1)^r & 0 & 0 & 1
        \end{pmatrix}.
    \end{align}
    \item $Y(z;r,a,b) = \big(I+ O(1/z)\big) \begin{pmatrix}
        z^{-N} & 0 & 0 & 0 \\ 0 & z^{N} & 0 & 0 \\
        0 & 0 & 1 & 0 \\
        0 & 0 & 0 & 1
    \end{pmatrix}$, as $z \rightarrow \infty$.
\end{enumerate}
\end{problem}
\begin{theorem}
Let $L,M,N$ be positive integers with $M<L$. Let $r,a,b$ be positive integers such that
$r \in \{1, \dots, L-1\}$ and $\max\{0,r+M-L\} \leq a < b \leq \min\{N+r,N+M\}$.
The kernel $K_{L,M,N}^{r,a,b}$ of the determinantal point process $\mathbb P_{L,M,N}^{r,a,b}$ is given by
    \begin{multline}
    K_{L,M,N}^{r,a,b}(t,h;t',h') = -1_{t>t'} \oint_{\Sigma_1} z^{-h} \frac{W_{t'}(z)}{W_t(z)} z^{h'} \frac{dz}{2\pi iz} \\
    + \int_{\Sigma_1} \int_{\Sigma_1} u^{-h} (u+1)^{-L+t} \frac{v}{u-v} \vec \phi_t(u;r,a,b) Y_\pm(u;r,a,b)^{-1} \\ \times Y_\pm(v;r,a,b) \vec \psi_{t'}(v;r,a,b) (v+1)^{L-t'} v^{h'} \frac{dv}{2\pi iv} \frac{du}{2\pi iu},
\end{multline}
where
\begin{equation}
    \vec \phi_t(u;r,a,b) := \begin{pmatrix}
        (u+1)^L \\ 0 \\ -1_{t>r} u^b(u+1)^{L-r} \\ 1_{t>r} u^a(u+1)^{L-r}
    \end{pmatrix}, \quad \vec \psi_{t'}(v;r,a,b) := \begin{pmatrix}
        0 \\ v^{-(M+N)} \\ - 1_{r>t'} v^{-b}(v+1)^{-L+r} \\ - 1_{r>t'} v^{-a}(v+1)^{-L+r}
    \end{pmatrix},
\end{equation}
and where $Y$ solves RH problem \ref{prbl: hexagon}.
\end{theorem}
\begin{proof}
    The statement follows from a direct application of Theorem \ref{thm: main} with $K$ given by \eqref{def: kernelhexagon} and $B$ given by \eqref{def: gapH}. By \eqref{def: J1}--\eqref{def: J2}, the jump matrices $J^{(1)}(\cdot;r,a,b)$ and $J^{(2)}(\cdot;r,a,b)$ are given by \eqref{J1hexagon}--\eqref{J2hexagon}.
    This completes the proof.
\end{proof}
\begin{remark}\label{remark:choiceY}
One could also apply Theorem \ref{thm: main} using the representation \eqref{def: kernelRHP} with the canonical choice $Y=I_2$. This would lead to a RH problem with more complicated jump matrices involving the solution of the model RH problem \ref{prbl: hexagon}.  
One can use a dressing procedure similar to the one in \cite[Section 6]{charlier2025countingdominolozengetilings} to show that the resulting RH problem is equivalent to our RH problem \ref{prbl: hexagon2}, but nevertheless this shows the convenience of choosing $Y$ different from the identity matrix in \eqref{def: kernelRHP}.
\end{remark}
\begin{remark}
    One can produce more complicated reduced hexagons by conditioning on voids of points in multiple clusters of the form \eqref{def: gapH}. This would result in multiple holes being removed from the domain. Repeatedly iterating the above argument would again allow us to express the correlation kernel as a double contour integral. The associated RH problem would have size given by twice the number of holes plus two.
\end{remark}

\section*{Acknowledgments}
No generative AI tools have been used to produce the mathematical content or the presentation of this manuscript.
The authors are grateful to Christophe Charlier and Arno Kuijlaars for useful discussions. TC
acknowledges support by FNRS Research Project T.0028.23 and by the Fonds Sp\'ecial de Recherche
of UCLouvain. FG acknowledges support by the Global PhD Partnership UCLouvain - KU Leuven.

\bibliography{bibliography.bib}

\begin{verbatim}
Tom Claeys 
Institute for Research in Mathematics and Physics, UCLouvain, Belgium
tom.claeys@uclouvain.be

Felix Gideonse
Institute for Research in Mathematics and Physics, UCLouvain, Belgium, and 
Department of Mathematics, KU Leuven, Belgium
felix.gideonse@uclouvain.be
\end{verbatim}

\end{document}